\documentclass[runningheads]{llncs}

\usepackage{amsmath}
\usepackage{amssymb}
\usepackage{graphicx}

\usepackage{xcolor}
\usepackage{fancyvrb}
\usepackage[all]{xy}
\SelectTips{cm}{}
\usepackage{xspace}
\usepackage{refcount}

\newcommand{\rew}{\leadsto}
\newcommand{\rewstar}{\rew^{*}}
\newcommand{\rel}{gql-narrowing}
\newcommand{\hsp}{\null\hspace{1pc}}
\newcommand{\To}{\Rightarrow}
\newcommand{\lwedge}{\,\wedge\,} 

\newcommand{\syst}{\mathcal{R}_{{\it gql}}}
\newcommand{\gqa} {\mathcal{G}\!\mathcal{Q}}

\newcommand{\conselect}{CONSELECT}
\newcommand{\Gr} {\texttt{Gr}\xspace}

\newcommand{\Pat} {\texttt{Pat}\xspace}
 
\newcommand{\configurations} {\texttt{Configurations}\xspace}
\newcommand{\Exp} {\texttt{Exp}\xspace}

\newcommand{\Som}{\texttt{Som}\xspace}

\newcommand{\Var}{\texttt{Var}\xspace}

\newcommand{\eval}[3]{\mathit{ev}(#1,#2)_{#3}}
\newcommand{\ev}{\underline{\mathit{ev}}}
\newcommand{\uleval}[2]{\underline{\mathit{ev}}(#1,#2)}

\newcommand{\sparql}{\mathrm{SPARQL}\xspace}

\newcommand{\clause}[1]{\textrm{#1}} 

\newcommand{\patalgebra}{\mathcal{P}_{{\it gql}}}
\newcommand{\Ll}{{\mathcal{L}}} 
\newcommand{\mC}{{\mathcal{C}}} 
\newcommand{\V}{{\mathcal{V}}}

\newcommand{\ulm}{{\underline{m}}}  
\newcommand{\ulp}{{\underline{p}}}
\newcommand{\ulc}{{\underline{c}}}

\newcommand{\uli}{{\underline{i}\,}}

\newcommand{\emptysetbis}{\varnothing} 
\newcommand{\emptygraph}{\emptyset} 
\newcommand{\emptypattern}{\square}

\newcommand{\config}[3]{{[}#2,#1{]}} 
\newcommand{\configdeux}[2]{{[}#2,#1{]}} 
\newcommand{\conf}[2]{{[}\,#1,\,#2\,{]}}

\newcommand{\Ulm}{\mathit{Match}} 
\newcommand{\tbl}{\mathit{Tab}}
 
\newcommand{\op}{\mathit{op}} 
\newcommand{\Op}{\mathit{Op}} 
\newcommand{\agg}{\mathit{agg}} 
\newcommand{\Agg}{\mathit{Agg}} 
\newcommand{\gp}{\mathit{gp}} 
\newcommand{\true}{\mathit{true}} 
\newcommand{\false}{\mathit{false}} 
\newcommand{\elem}{\mathit{elem}}
\newcommand{\ex}{\mathit{ex}}

\newcommand{\opn}[1]{\textit{#1}} 
\newcommand{\Solve}{\textit{Solve}} 

\newcommand{\expr}{e} 
\newcommand{\deno}[1]{[[#1]]} 
\newcommand{\bag}[1]{\{\!|#1|\!\}} 

\newcommand{\se}[1]{{[#1]}} 
\newcommand{\sem}[2]{{[[#1]]_{#2}}} 
\newcommand{\seg}[2]{{{#2}^{(#1)}}}

\newcommand{\Result}{\mathit{Result}}
\newcommand{\Print}{\mathit{Print}}  
\newcommand{\Graph}{\mathit{Gr}}  
\newcommand{\param}{\mathit{param}}
\newcommand{\var}{\mathit{var}}

\newcommand{\hide}[1]{}

\begin{document}
\title{A Rule-based Operational Semantics of Graph
  Query Languages
  \thanks{Partly supported by the French ANR project VERIGRAPH \#~ANR-21-CE48-0015}}
\author{Dominique Duval\inst{1}\orcidID{ 0000-0002-4080-2783} \and
Rachid Echahed\inst{2}\orcidID{0000-0002-8535-8057} \and
Fr\'ed\'eric Prost\inst{2}\orcidID{0000-0001-6947-4819}}
\authorrunning{D. Duval, R. Echahed \& F. Prost}
\titlerunning{A Rule-based Operational Semantics of Graph
  Query Languages.}

\institute{LJK - Univ. Grenoble Alpes, Grenoble, France\\
\email{dominique.duval@univ-grenoble-alpes.fr}\\
 \and
LIG - CNRS and Univ. Grenoble Alpes, Grenoble, France\\
\email{rachid.echahed@imag.fr,\ \ frederic.prost@univ-grenoble-alpes.fr}}
\maketitle              

\begin{abstract}

    We consider a core language of graph queries. These queries are seen
  as formulas to be solved with respect to graph-oriented databases.
  For this purpose, we first define a graph query algebra where some
  operations over graphs and sets of graph homomorphisms are
  specified. Then, the notion of pattern is introduced to represent a
  kind of recursively defined formula over graphs. The syntax and
  formal semantics of patterns are provided. Afterwards, we propose a
  new sound and complete calculus to solve patterns. This calculus,
  which is based on a rewriting system, develops only one derivation
  per pattern to be solved. Our calculus is generic in the sense that
  it can be adapted to different kinds of graph databases provided
  that the notions of graph and graph homomorphism (match) are well
  defined.

  \keywords{Operational semantics, Rewrite systems, Graph query languages}
\end{abstract}

\section{Introduction}

Rewriting techniques have been widely used in different areas such as
operational semantics of declarative languages or automated theorem
proving. In this paper, our main aim is to propose to use such
techniques in the case of graph-oriented database languages.

Current developments in database theory show a clear shift from
relational to graph-oriented databases. Relational databases are now
well mastered and have been largely investigated in the literature
with an ISO standard language SQL \cite{Chamberlin1974,Date:1987:GSS}.
On the other side, the wide use of graphs as a flexible data model for
numerous database applications \cite{SakrBVIAAAABBDV21} as well as the
emergence of various languages such as SPARQL \cite{sparql}, Cypher
\cite{cypher} or G-CORE \cite{AnglesABBFGLPPS18} to quote a few. An
ongoing ISO project of a standard language, called GQL, has emerged
recently for graph-oriented databases
\footnote{https://www.gqlstandards.org/}.

Representing data graphically is quite legible. However, there is
always a dilemma in choosing the right notion of graphs when modeling
applications. This issue is already present in some well investigated
domains such as modeling languages \cite{BorkKP20} or graph transformation
\cite{handbook1}. Graph-oriented data representation does not escape from such
dilemma. We can quote for example RDF graphs \cite{rdf} on which SPARQL
is based or Property Graphs \cite{cypher} currently used in several
languages such as Cypher, G-CORE or the forthcoming GQL language.

In addition to the possibility of using different graph representations
for data, graph database languages feature new kinds of queries such as
graph-to-graph queries, cf. CONSTRUCT queries in SPARQL or G-CORE, besides the
classical graph-to-relation (table) queries such as SELECT or MATCH
queries in SPARQL or Cypher. The former constitute a class of queries
which transforms a graph database to another graph database. The
later transforms a graph to a multiset 
of solutions represented in general
by means of a table 
just as in the classical relational framework.

In general, graph querying processing integrates features shared with
graph transformation techniques (database transformation) and goal
solving (variable assignments). Our main aim in this paper is to
define an operational semantics, based on rewriting techniques, for 
graph-oriented queries. We propose a generic rule-based calculus, called
\emph{\rel}\ which is parameterized by the actual interpretations
of graphs and their matches (homomorphisms).  That is to say, the
obtained calculus can be adapted to different definitions of graph
and the corresponding notion of match. The
proposed calculus consists on a dedicated rewriting system and a
narrowing-like procedure which follows closely the formal semantics of patterns or
queries, the same way as (SLD-)Resolution calculus is related to
formal models underlying Horn or Datalog clauses. The use of rewriting
techniques in defining the proposed operational semantics paves the
way to syntactic analysis and automated verification techniques for
the proposed core language.

In order to define a sound and complete calculus, we first propose a
uniform formal semantics for queries. Actually, we do consider
graph-to-graph queries and graph-to-table queries as two facets of one
same syntactic object that we call \emph{pattern}. The semantics of a
pattern is a set of matches, that is to say, a set of graph
homomorphisms and not only a set of variable assignments as proposed
in \cite{AnglesABHRV17,cypher}. 
From such set of matches, one can easily display either the
tables by considering the images of the variables as defined by the
matches or the graph target of the matches or even
both tables and graphs. Our semantics for patterns allows us to write
nested patterns in a natural way, that is, new data graphs can be
constructed on the fly before being queried.

The paper is organized as follows: next section introduces a graph
query algebra featuring some key operations needed to express the
proposed calculus. Section~\ref{sec:gral} defines the syntax of
patterns and queries as well as their formal semantics. In
Section~\ref{sec:opsemantics}, a sound and complete calculus is
given. First we introduce a rewriting system describing how query
results 
are found. Then, we define \rel, which is associated
with the proposed rules. Concluding remarks and related work are given
in Section~\ref{sec:conclusion}.

\section{Graph Query Algebra}
\label{sec:algebra}

During a query answering process, different intermediate results 
can be computed and composed. In this section, we introduce a Graph Query
Algebra $\gqa$ which consists of a family of operations over graphs,
matches (graph homomorphisms) and expressions.  These
operations are used later on to define the semantics of queries, see
Sections~\ref{sec:gral} and~\ref{sec:opsemantics}.

\subsection{Signature for the Graph Query Algebra}
\label{ssec:algebra-signature}

The algebra $\gqa$ is defined over a signature.
The main sorts of
this signature are \Gr, \Som, \Exp and \Var to be interpreted as
graphs, sets of matches, expressions and variables, respectively, as
explained in Sections~\ref{ssec:algebra-graphs},
\ref{sec:more_on_matches}, \ref{ssec:gral-expr}
and~\ref{ssec:algebra-operation}.  The sort \Var is a subsort of \Exp.
The main operators of the signature are:   
\begin{itemize}
\item $\opn{Match} : \Gr, \Gr \to \Som $
\item $\opn{Join} : \Som, \Som \to \Som$
\item $\opn{Bind} : \Som, \Exp, \Var \to \Som$
\item $\opn{Filter} : \Som, \Exp \to \Som$
\item $\opn{Build} : \Som, \Gr \to \Som$
\item $\opn{Union} : \Som, \Som \to \Som$
\end{itemize}

The above sorts and operations are given as an
indication while being inspired by concrete languages. They may be
modified or adapted according to actual graph-oriented query languages.

\subsection{An Actual Interpretation of Graphs}
\label{ssec:algebra-graphs}

Various interpretations of sorts $\Gr$ and  $\Som$ can be given.
In order to provide concrete examples, we have to fix an actual
interpretation of these  sorts. For all the examples given in the paper,
we have chosen to interpret the sort $\Gr$ as generalized RDF graphs.
We could of course have chosen other notions of graphs such as
property graphs. Our choice here is motivated by the simplicity of RDF
graph definition (set of triples).

Below, we define generalized RDF graphs.They are the usual RDF graphs but
they may contain isolated nodes. 
Let $\Ll$ be a set, called the set of \emph{labels}, 
made of the union of two disjoint sets $\mC$ and $\V$, called respectively
the set of \emph{constants} and the set of \emph{variables}.

\begin{definition}[graph] 
\label{def:algebra-graph}
Every element $t=(s,p,o)$ of $\Ll^3$ is called a \emph{triple}  
and its members $s$, $p$ and $o$ are called respectively 
the \emph{subject}, the \emph{predicate} and the \emph{object} of $t$.
A \emph{graph} $G$ is a pair $G=(G_N,G_T)$ made of a subset $G_N$ of $\Ll$
called the set of \emph{nodes} of $G$ and a subset $G_T$ of $\Ll^3$ 
called the set of \emph{triples} of $G$, such that the subject and the object
of each triple of $G$ are nodes of $G$.
The nodes of $G$ which are neither a subject nor an object are called
the \emph{isolated nodes} of $G$.
The set of \emph{labels} of a graph $G$ is the subset $\Ll(G)$
of $\Ll$ made of the nodes and predicates of $G$, then 
$\mC(G)=\mC\cap\Ll(G)$ and $\V(G)=\V\cap\Ll(G)$.
The graph with an empty set of nodes and an empty set of triples
is called the \emph{empty graph} and is denoted by $\emptygraph$.
Given two graphs $G_1$ and $G_2$, 
the graph $G_1$ is a \emph{subgraph} of $G_2$, written $G_1 \subseteq G_2$,
if $(G_1)_N \subseteq (G_2)_N$ and $(G_1)_T \subseteq (G_2)_T$,
then $\Ll(G_1) \subseteq \Ll(G_2)$.
The \emph{union} $G_1\cup G_2$ is the graph 
defined by $(G_1\cup G_2)_N = (G_1)_N \cup (G_2)_N$ and
$(G_1\cup G_2)_T = (G_1)_T \cup (G_2)_T$,
then $\Ll(G_1\cup G_2)=\Ll(G_1)\cup \Ll(G_2)$.
\end{definition}

In the rest of the paper we write graphs as a couple made of a set of triples and a set of  nodes:
for example the graph $G = (\{n_1,n_2,s_1,p_1\} \,,\, \{(s_1,o_1,p_1)\})$
which is made of four nodes and one triple is written as
$G = \{(s_1,o_1,p_1),n_1, n_2\}$.
\begin{example}
\label{ex:database}

We define a toy database which is used as a running example throughout the paper.
The database consists of {\em persons} who {\em are} either {\em professors} or {\em students},
with {\em topics}
such that each professor {\em teaches} some topics 
and each student {\em studies} some topics.
\begin{center}
\begin{tabular}{llll} 
    $G_\ex= \{$ & (Alice, is, Professor), & (Alice, teaches, Mathematics), & \\
    & (Bob, is, Professor), & (Bob, teaches, Informatics), & \\
    & (Charlie, is, Student), & (Charlie, studies, Mathematics), & \\
    & (David, is, Student), & (David, studies, Mathematics), & \\ 
    & (Eric, is, Student), & (Eric, studies, Informatics) & $ \}$\\
  \end{tabular}
\end{center}

\end{example}

Below, we define the notion of \emph{match} which will be used, notably, 
to represent results of queries.

\begin{definition}[match]
\label{def:algebra-match}
A \emph{graph homomorphism} from a graph $L$ to a graph $G$,
denoted $m:L\to G$, is a function from $\Ll(L)$ to $\Ll(G)$
which \emph{preserves nodes} and \emph{preserves triples}, in the sense that 
$m(L_N)\subseteq G_N$ and $m^3(L_T)\subseteq G_T$. 
A \emph{match} is a graph homomorphism $m:L\to G$ 
which \emph{fixes} $\mC$, in the sense that 
$m(c)=c$ for each $c$ in $\mC(L)$.
\end{definition}

When $n$ is an isolated node of $L$ then the node $m(n)$ does not
have to be isolated in $G$.
A match $m:L\to G$ determines two functions 
$m_N:L_N\to G_N$ and $m_T:L_T\to G_T$, restrictions of $m$ and $m^3$ 
respectively. 
A match $m:L\to G$ is 
invertible if and only if 
both functions $m_N$ and $m_T$ are bijections.
This means that a function $m$ from $\Ll(L)$ to $\Ll(G)$
is an 
invertible match if and only if
$\mC(L)=\mC(G)$ with $m(c)=c$ for each $c\in\mC(L)$
and $m$ is a bijection from $\V(L)$ to $\V(G)$:
thus, $L$ is the same as $G$ up to variable renaming.
It follows that the symbol used for naming a variable does not matter
as long as graphs are considered only up to invertible matches.

Notice that RDF graphs \cite{rdf} are graphs according to
Definition~\ref{def:algebra-graph} but 
without isolated nodes, and where constants are either IRIs
(Internationalized Resource Identifiers) or literals
and where all predicates are IRIs and only objects can be literals.
Blank nodes in RDF graphs are the same as variable nodes in our graphs. 
An isomorphism of RDF graphs, as defined in~\cite{rdf},
is an invertible match. 
isomorphism of graphs as in Definition~\ref{def:algebra-match}.

\subsection{More Definitions on Matches}
\label{sec:more_on_matches}

Below we introduce some useful definitions on matches. 
Notice that we do not consider a  match $m$ as a simple variable
assignment but rather as a graph homomorphism with a clear source and target
graphs. This nuance in the definition of matches is important in the
rest of the paper.

\begin{definition}[compatible matches] 
\label{def:algebra-compatible}  
Two matches $m_1:L_1\to G_1$ and $m_2:L_2\to G_2$ are \emph{compatible},
written as $m_1\sim m_2$, if $m_1(x)=m_2(x)$ for each $x\in\V(L_1)\cap\V(L_2)$. 
Given two compatible matches $m_1:L_1\to G_1$ and $m_2:L_2\to G_2$, 
let $m_1\bowtie m_2 : L_1\cup L_2 \to G_1\cup G_2$
denote the unique match such that $m_1\bowtie m_2\sim m_1$
and $m_1\bowtie m_2\sim m_2$ 
(which means that $m_1\bowtie m_2$ coincides with $m_1$ on $L_1$
and with $m_2$ on $L_2$). 
\end{definition}

\begin{definition}[building a match]
  \label{def:algebra-build}
Let $m: L \to G$ be a match and $R$ a graph. 
The match $\opn{Build}(m,R) : R\to G\cup H_{m,R}$
is the unique match (up to variable renaming)
such that for each variable $x$ in $R$: 
$$ \opn{Build}(m,R)(x) =
\begin{cases}
  m(x) \mbox{ when } x \in \V(R)\cap\V(L), \\
  \mbox{some fresh variable } \var(m,x) \mbox{ when } x \in \V(R) -\V(L). \\
\end{cases}$$
and $H_{m,R}$ is the image of $R$ by $\opn{Build}(m,R)$.
\end{definition}

\begin{definition}[set of matches, assignment table] 
\label{def:algebra-table}
Let $L$ and $G$ be graphs.
A set $\ulm$ of matches, all of them from $L$ to $G$,
is denoted $\ulm:L\To G$ and called a \emph{homogeneous set of matches},
or simply a \emph{set of matches},
with \emph{source} $L$ and \emph{target} $G$.
The \emph{image} of $L$ by $\ulm$ is the subgraph
$\ulm(L)=\cup_{m\in\ulm}(m(L))$ of $G$.
We denote $\opn{Match}(L,G):L\To G$ the set of all matches from $L$ to $G$. 
When $L$ is the empty graph this set has one unique element which 
 is the inclusion of $\emptygraph$ into $G$, then we denote
$\uli_G=\opn{Match}(\emptygraph,G):\emptygraph\To G$ this one-element set 
and $\emptysetbis_G:\emptygraph\To G$ its empty subset.
The \emph{assignment table} $\tbl(\ulm)$ of $\ulm$ is the two-dimensional
table with 
the elements of $\V(L)$ in its first row, then one row for each $m$ in $\ulm$, 
and the entry in row $m$ and column $x$ equals to $m(x)$.
\end{definition}

Thus, the assignment table $\tbl(\ulm)$ describes the set of functions
$\ulm|_{\V(L)}:\V(L)\To\Ll$,
made of the functions $m|_{\V(L)}:\V(L)\to\Ll$ for all $m\in\ulm$.
A set of matches $\ulm:L \To G$ is determined by the graphs $L$ and $G$
and the assignment table $\tbl(\ulm)$.

\begin{example}
  \label{ex:set-of-matches}

In order to determine when professor $?p$ teaches topic $?t$ 
which is studied by student~$?s$ 
we may consider the following graph $ L_\ex$,
where $?p$, $?t$ and $?s$ are variables. In all examples,
variables are preceded by a ``?''.

\begin{center}
  $ L_\ex=\{$ ($?p$, teaches, $?t$), ($?s$, studies, $?t$) $\} $
\end{center}
There are 3 matches from $L_\ex$ to $G_\ex$. 
The set $\ulm_\ex$ of all these matches is: 
 $$ \ulm_\ex:L_\ex \To G_\ex \;\mbox{ with }\;
   \tbl(\ulm_\ex) = 
\mbox { \begin{tabular}{|l|l|l|}
     \hline
     \multicolumn{1}{|c|}{$?p$} &
     \multicolumn{1}{|c|}{$?t$} & 
     \multicolumn{1}{|c|}{$?s$} \\
     \hline
       Alice & Mathematics & Charlie \\
       Alice & Mathematics & David \\
       Bob  & Informatics & Eric \\ 
   \hline
  \end{tabular}
  }
   $$

\end{example}

\subsection{Expressions} 
\label{ssec:gral-expr}

Query languages usually feature a term algebra dedicated to
express operations over integers, booleans and so forth.
We do not care here about the way basic operations are chosen  
but we want to deal with aggregation operations as in most 
database query languages.
 Thus, one can think of any kind of term algebra with
 operators which are classified as either basic operators (unary or
 binary)
and aggregation operators (always unary). 
We consider that all expressions are well typed.
Typically, and not exclusively, the sets $\Op_1$, $\Op_2$ and $\Agg$
of \emph{basic unary} operators, \emph{basic binary} operators 
and \emph{aggregation} operators can be: 
\\ \hsp $ \Op_1=\{-,\mathrm{NOT}\}\,,$
\\ \hsp $ \Op_2=\{+,-,\times,/,=,>,<,\mathrm{AND},\mathrm{OR}\}\,,$
\\ \hsp $\Agg = \Agg_{\elem} \cup \{\agg\;\mathrm{DISTINCT} \mid \agg\in\Agg_{\elem}\} $
\\ \hsp\hsp where $\Agg_{\elem} =
\{\mathrm{MAX},\mathrm{MIN},\mathrm{SUM},\mathrm{AVG},\mathrm{COUNT}\}$. 
\\ A \emph{group of expressions} is a non-empty finite list of expressions. 

\begin{definition}[syntax of expressions]\ 
\label{def:gral-syn-expr} 
\emph{Expressions} $\expr$ and their sets of \emph{in-scope}
variables $\V(\expr)$ are defined recursively as follows,  with
$c\in\mC$,  $x\in\V$, $\op_1\in\Op_1$, $\op_2\in\Op_2$, $\agg\in\Agg$,
$\gp$ is a group of expressions:

$\expr ::= c \mid x \mid \op_1\;\expr \mid \expr \;\op_2\;\expr \mid
\agg(\expr_1) \mid \agg(\expr_1\mbox{ {\rm BY} }\gp).$

 $\V(c)=\emptysetbis$, $\V(x)=\{x\}$, $\V(\op_1\;\expr)=\V(\expr)$,
 $\V(\expr_1\;\op_2\;\expr_2)=\V(\expr_1)\cup\V(\expr_2)$,
 $\V(\agg(\expr))=\V(\expr)$,

$\V(\agg(\expr \mbox{ {\rm BY} }\gp))=\V(\expr)$ 
(the variables in $\gp$ must be distinct from those in $\expr$).
 \end{definition}

The \emph{value} of an expression with respect to a set of matches $\ulm$ 
(Definition~\ref{def:gral-sem-expr}) is a family of constants 
$\uleval{\ulm}{\expr}=(\eval{\ulm}{\expr}{m})_{m\in\ulm}$ indexed by
the set $\ulm$. 
When the expression $\expr$ is free from any aggregation operator
then $\eval{\ulm}{\expr}{m}$ is simply $m(\expr)$. 
But in general $\eval{\ulm}{\expr}{m}$ depends on $\expr$ and $m$ 
and it may also depend on other matches in $\ulm$
when $\expr$ involves aggregation operators.
The \emph{value} of a group of expressions $\gp=(\expr_1,...,\expr_k)$ 
with respect to $\ulm$ is the list 
$\ev(\ulm,\gp)=(\ev(\ulm,\expr_1)_m,...,\ev(\ulm,\expr_k)_m)_{m \in \ulm}$. 
To each basic operator $\op$ is associated a function $\deno{op}$
(or simply $\op$) 
from constants to constants if $\op$ is unary
and from pairs of constants to constants if $\op$ is binary.
To each aggregation operator $\agg$ in $\Agg$ is associated a function
$\deno{agg}$ (or simply $\agg$)
from \emph{multisets} of constants to constants. 
Note that each family of constants determines a multiset of constants: 
for instance a family $\ulc=(c_m)_{m\in\ulm}$ of constants indexed by the
elements of a set of matches $\ulm$ determines the multiset of constants
$\bag{c_m\mid m\in\ulm}$, 
which is also denoted $\ulc$ when there is no ambiguity.
Some aggregation operators $\agg$ in $\Agg_{\elem}$ are such that 
$\deno{agg}(\ulc)$ depends only on the set underlying the multiset $\ulc$, 
which means that $\deno{agg}(\ulc)$ 
does not depend on the multiplicities in the multiset $\ulc$:
this is the case for MAX and MIN but not for SUM, AVG and COUNT.
When $\agg=\agg_{\elem}\;\mathrm{DISTINCT}$ with $\agg_{\elem}$ in $\Agg_{\elem}$
then $\deno{agg}(\ulc)$ is $\deno{agg_{\elem}}$ applied to
the underlying set of $\ulc$.
For instance, $\clause{COUNT }(\ulc)$ counts the number of elements
of the multiset $\ulc$ with their multiplicies, while 
$\clause{COUNT DISTINCT }(\ulc)$ counts the number of distinct elements
in~$\ulc$. 

\begin{definition}[evaluation of expressions]
\label{def:gral-sem-expr}
Let $L$ be a graph, $\expr$ an expression over $L$ and $\ulm:L\To G$
a set of matches. 
The \emph{value} of $\expr$ with respect to $\ulm$ is the family 
$\uleval{\ulm}{\expr} = (\eval{\ulm}{\expr}{m})_{m\in\ulm}$ defined recursively
as follows. It is assumed that each $\eval{\ulm}{\expr}{m}$ in this definition
is a constant. 
\begin{itemize}
  \item $\eval{\ulm}{c}{m} = c$,
  \item $\eval{\ulm}{x}{m} = m(x)$,
\item $\eval{\ulm}{\op\;\expr_1}{m} = \deno{\op}\,\eval{\ulm}{\expr_1}{m}\,$,  
\item $\eval{\ulm}{\expr_1\;\op\;\expr_2}{m} =
  \eval{\ulm}{\expr_1}{m}\,\deno{\op}\,\eval{\ulm}{\expr_2}{m}\,$,
\item $\eval{\ulm}{\agg(\expr_1)}{m} = \deno{\agg}(\ev(\ulm,\expr_1))$,
\item $\eval{\ulm}{\agg(\expr_1\;BY\;\gp)}{m} =
  \deno{\agg}(\ev(\ulm|_{\gp,m},\expr_1))$ 
  where $\ulm|_{\gp,m}$ is the subset of $\ulm$ 
  made of the matches $m'$ in $\ulm$ such that
  $ \ev(\ulm,\gp)_{m'} = \ev(\ulm,\gp)_m$.
\end{itemize}
Note that $\eval{\ulm}{\agg(\expr_1)}{m}$ is the same for all $m$ in $\ulm$
while $\eval{\ulm}{\agg(\expr_1\;BY\;\gp)}{m}$ is the same for all $m$
and $m'$ in $\ulm$ such that $\eval{\ulm}{\gp}{m}=\eval{\ulm}{\gp}{m'}$.

\end{definition} 

\subsection{Operations}
\label{ssec:algebra-operation}

The sorts $\Gr$, $\Som$, $\Exp$ and $\Var$ 
of the signature in Section~\ref{ssec:algebra-signature} are
interpreted in the algebra $\gqa$ respectively as 
the set of graphs (Definition~\ref{def:algebra-graph}),
the set of homogeneous sets of matches (Definition~\ref{def:algebra-table}), 
the set of expressions (Definition~\ref{def:gral-syn-expr})
and its subset of variables. 
Then the operators of the signature are interpreted in the algebra $\gqa$
by the operations with the same name in Definition~\ref{def:algebra-opn}.
Whenever needed, we extend the target of matches:
for every graph $H$ and every match $m:L\to G$ where $G$ is a subgraph of $H$
we denote $m:L\to H$ when $m$ is considered as a match from $L$ to $H$.

\begin{definition}[$\gqa$ operations]\
\label{def:algebra-opn}
\begin{itemize}
\item 
For all graphs $L$ and $G$: 
\\ \hsp $\opn{Match}(L,G):L\To G $ is the set of all matches from $L$ to $G$.
\item 
For all sets of matches $\ulm:L\To G$ and $\ulp:R\To H$: 
\\ \hsp $\opn{Join}(\ulm,\ulp)
= \{ m\bowtie p \mid m\in\ulm \lwedge p\in\ulp \lwedge m\sim p\}
: L \cup R \To G\cup H$.
\item
  For every set of matches $\ulm:L\To G$,
every expression $e$ and every variable~$x$,
let $p_m(x)=\eval{\ulm}{e}{m}$ for each $m\in\ulm$. Then: 
\\ \hsp $\opn{Bind}(\ulm,e,x) \!=\! 
\{ m\!\bowtie\! p_m \mid m\!\in\!\ulm \!\lwedge\! m\!\sim\! p_m\} \!:\!
L\cup\{x\} \!\To\! G\cup\{p_m(x) \mid m\!\in\!\ulm \} $.
\\ Equivalently, this can be expressed as follows:
\\ if $x\in\V(L)$ then $\opn{Bind}(\ulm,e,x) =
\{ m \mid m\in\ulm \lwedge m(x)=p_m(x) \}:L\To G$, 
\\ otherwise
$\opn{Bind}(\ulm,e,x)\!=\!\{ m\bowtie p_m \mid m\!\in\!\ulm \} \!:\!
L\cup\{x\} \!\To\! G\cup\{p_m(x) \mid m\!\in\!\ulm \} $.
\item For every set of matches $\ulm:L\To G$ and every expression $e$: 
\\ \hsp $\opn{Filter}(\ulm,e) = 
\{ m \mid m\in\ulm \lwedge \eval{\ulm}{\expr}{m} = \true \}:L\To G$.
\item For every set of matches $\ulm:L\To G$ and every graph $R$: 
\\ \hsp $\opn{Build}(\ulm,R) =
\{ \opn{Build}(m,R) \mid m\in\ulm \}:R\To G \cup \opn{Build}(\ulm,R)(R) $
\\ \hsp where $\opn{Build}(\ulm,R)(R)=\cup_{m\in\ulm} \opn{Build}(m,R)(R) $.
\item For all sets of matches $\ulm:L\To G$ and $\ulp:L\To H$:
\\ \hsp $\opn{Union}(\ulm,\ulp)
= (\ulm : L \To G\cup H) \;\cup\; (\ulp : L \To G\cup H)
: L \To G\cup H $.
\end{itemize}
\end{definition}

\section{Patterns and Queries} 
\label{sec:gral}

Syntax of graph-oriented dabases is still evolving. We do not consider
all technical syntactic details of a real-world language nor all
possible constraints on matches. We focus on a core language. Its
syntax reflects significant aspects of graph-oriented
queries. Conditions on graph paths, which can be seen as constraints
on matches, are omitted in this paper in order not to make the
syntax too cumbersome.
We consider mainly two syntactic categories: \emph{patterns} and
\emph{queries}, in addition to \emph{expressions} already mentioned in
Section~\ref{ssec:gral-expr}.
Queries are either $\clause{SELECT}$ queries, as in most query languages, 
$\clause{CONSTRUCT}$ queries, as in $\sparql$ and G-CORE, or the new
$\clause{\conselect}$ queries introduced in this paper.
A $\clause{SELECT}$ query applied to a graph returns a \emph{table}
which describes a multiset of \emph{solutions} or variable bindings,
while a $\clause{CONSTRUCT}$ query applied to a graph returns a graph.
A $\clause{\conselect}$ query applied to a graph returns both a graph
and a table.
On the other hand, a pattern applied to a graph returns a set of matches.
Patterns are the basic blocks for building queries.
They are defined in Section~\ref{ssec:gral-pattern}
together with their semantics.
Queries are defined in Section~\ref{ssec:gral-query} and 
their semantics is easily derived from the semantics of patterns.
In this Section, as in Section~\ref{sec:algebra},  
the set of \emph{labels} $\Ll$ is the union of the disjoint sets
$\mC$ and $\V$, of \emph{constants} and \emph{variables} respectively.
We assume that the set $\mC$ of constants contains the numbers
and strings and the boolean values $\true$ and $\false$.

\subsection{Patterns}
\label{ssec:gral-pattern}

In Definition~\ref{def:gral-syn-pattern} patterns are built
from graphs by using six operators: 
$\clause{BASIC}$, $\clause{JOIN}$, $\clause{BIND}$, $\clause{FILTER}$,
$\clause{BUILD}$ and $\clause{UNION}$.
Then, in Definition~\ref{def:gral-sem-pattern} the formal semantics of patterns
is given by an evaluation function.

\begin{definition}[syntax of patterns]
  \label{def:gral-syn-pattern}
  \emph{Patterns} $P$ and their \emph{scope graph} $\se{P}$ 
are defined recursively as follows.
\begin{itemize}
\item The symbol $\emptypattern$ is a pattern,
  called the \emph{empty pattern},
  and $\se{\emptypattern}$ is the empty graph $\emptygraph$. 
\item If $L$ is a graph then $P=\clause{BASIC}(L)$ is a pattern,
  called a \emph{basic pattern},
  and $\se{P} = L$. 
\item If $P_1$ and $P_2$ are patterns then
  $P=P_1\clause{ JOIN }P_2$ is a pattern
  and $\se{P} = \se{P_1}\cup\se{P_2}$.
\item If $P_1$ is a pattern, $\expr$ an expression such that
  $\V(\expr)\subseteq\V(\se{P_1})$ and $x$ a variable 
  then $P=P_1\clause{ BIND }\expr\clause{ AS }x$ is a pattern
  and $\se{P} = \se{P_1}\cup\{x\}$. 
  \item If $P_1$ is a pattern and $\expr$ an expression such that
  $\V(\expr)\subseteq\V(\se{P_1})$ 
  then $P=P_1\clause{ FILTER }\expr$ is a pattern 
  and $\se{P} = \se{P_1}$.
\item If $P_1$ is a pattern and $R$ a graph 
  then $P=P_1\clause{ BUILD }R$ is a pattern  
  and $\se{P} = R$.
\item If $P_1$ and $P_2$ are patterns such that $\se{P_1} = \se{P_2}$ 
  then $P=P_1\clause{ UNION }P_2$ is a pattern with 
   $\se{P} = \se{P_1} = \se{P_2}$.
\end{itemize}
\end{definition}

The \emph{value} of a pattern over a graph is a set of matches, as defined now.

\begin{definition}[evaluation of patterns, set of solutions]
\label{def:gral-sem-pattern}
The \emph{set of solutions} or the \emph{value} of a pattern $P$ over 
a graph $G$ is a set of matches $\sem{P}{G}:\se{P}\To\seg{P}{G}$
from the scope graph $\se{P}$ of $P$ 
to a graph $\seg{P}{G}$ that contains $G$.
This value $\sem{P}{G}:\se{P}\To\seg{P}{G}$ is 
defined inductively as follows:
\begin{itemize}
\item  $\sem{\emptypattern}{G}
  = \emptysetbis_G:\emptygraph\To G $.
  \item
    $\sem{BASIC(L)}{G}
  = \opn{Match}(L,G)
  :L\To G $.
\item $\sem{P_1\clause{ JOIN }P_2}{G}
  = \opn{Join}(\sem{P_1}{G},\sem{P_2}{\seg{P_1}{G}})
  : \se{P_1}\cup\se{P_2}\To \seg{P_2}{\seg{P_1}{G}}$.
\item $\sem{P_1\clause{ BIND }\expr\clause{ AS }x}{G}
  = \opn{Bind}(\sem{P_1}{G},\expr,x)
  : \se{P_1}\cup\{x\}\To \seg{P_1}{G}\cup{\sem{P_1}{G}(\expr)} $.
\item $\sem{P_1\clause{ FILTER }\expr}{G}
  = \opn{Filter}(\sem{P_1}{G},\expr)
  : \se{P_1}\To \seg{P_1}{G} $.
\item $\sem{P_1\clause{ BUILD }R}{G}
  = \opn{Build}(\sem{P_1}{G},R)
  : R\To \seg{P_1}{G} \cup \sem{P_1}{G}(R)$. 
\item $\sem{P_1\clause{ UNION }P_2}{G} 
  = \opn{Union}(\sem{P_1}{G},\sem{P_2}{\seg{P_1}{G}})
  : \se{P_1} \To \seg{P_2}{\seg{P_1}{G}} $. 
\end{itemize}
\end{definition}

\begin{remark} 
  \label{rem:gral-sem-pattern}
In all cases, the graph $\seg{P}{G}$ is built by adding to $G$
``whatever is required'' for the evaluation.  
When $P$ is the empty pattern, the value of $P$ over $G$
is the empty subset of $\opn{Match}(\emptygraph,G)$.
Syntactically, each operator $\clause{OP}$ builds a pattern $P$
from a pattern $P_1$ and a parameter $\param$,
which is either a pattern $P_2$ (for $\clause{JOIN}$ and $\clause{UNION}$),
a pair $(\expr,x)$ made of an expression and a variable (for $\clause{BIND}$),
an expression $\expr$ (for $\clause{FILTER}$)
or a graph $R$ (for $\clause{BUILD}$).
Semantically, for every pattern $P=P_1\clause{ OP }\param$,
let us denote $\ulm_1:X_1\To G_1$ for $\sem{P_1}{G}:\se{P_1}\To\seg{P_1}{G}$
and $\ulm:X\To G'$ for $\sem{P}{G}:\se{P}\To\seg{P}{G}$. 
In every case it is necessary to evaluate $\ulm_1$ before evaluating $\param$:
for $\clause{JOIN}$ and $\clause{UNION}$ this is because pattern $P_2$
is evaluated on $G_1$,
for $\clause{BIND}$ and $\clause{FILTER}$ because expression $\expr$
is evaluated with respect to $\ulm_1$,
and for $\clause{BUILD}$ because of the definition of $\opn{Build}$.
Note that the semantics of $P_1\clause{ JOIN }P_2$ and $P_1\clause{ UNION }P_2$
is not symmetric in $P_1$ and $P_2$ in general, unless $\seg{P_1}{G}=G$
and $\seg{P_2}{G}=G$, which occurs when $P_1$ and $P_2$ are basic patterns.
Given a pattern $P=P_1\clause{ OP }\param$, the pattern 
$P_1$ is a \emph{subpattern} of $P$, as well as $P_2$ when
$P=P_1\clause{ JOIN }P_2$ or $P=P_1\clause{ UNION }P_2$.
The semantics of patterns is defined in terms of the semantics of
its subpatterns (and the semantics of its other arguments, if any).
Thus, for instance, BUILD patterns can be nested at any depth.
\end{remark}

\begin{definition} 
\label{def:gral-sem-scope}
For every pattern $P$, the set $\V(P)$ of in-scope variables of $P$
is the set $\V(\se{P})$ of variables of the scope graph $\se{P}$.
An expression $\expr$ is \emph{over} a pattern $P$ if
$\V(\expr)\subseteq\V(P)$. 
\end{definition} 

\begin{example}
\label{ex:gral-sem-pattern}

Let $R_\ex$ be the following graph, where $?p$, $?z$ and $?s$ are variables.
\begin{center}
  $ R_\ex=\{$ ($?p$, teaches, $?z$), ($?s$, studies, $?z$) $\} $
\end{center}
Note that $R_\ex$ is the same as $L_\ex$, except for the name of one variable.
In order to determine when professor $?p$ teaches some topic 
which is studied by student~$?s$, whatever the topic, 
we consider the following pattern $P_\ex$.
\begin{center}
\begin{tabular}{lll} 
  $ P_\ex $ & $=$ & 
  $\clause{ BASIC }(L_\ex) \clause{ BUILD }R_\ex$ \\
  & $=$ & $\clause{ BASIC }(\{$ 
($?p$, \mbox{teaches}, $?t$), ($?s$, \mbox{studies}, $?t$) $ \}$) \\
  & & $\clause{ BUILD }\{$
($?p$, \mbox{teaches}, $?z$), ($?s$, \mbox{studies}, $?z$) $ \}$ \\
  \end{tabular}
\end{center}
Note that the variable $?z$ in $R_\ex$ does not appear in $L_\ex$.
Since there are 3 matches from $L_\ex$ to $G_\ex$
(Example~\ref{ex:set-of-matches}),
the value of $P_\ex$ over $G_\ex$ is: 
$$ \ulp_\ex:R_\ex \To G'_\ex \;\mbox{ with }\;
   \tbl(\ulp_\ex) = 
\mbox { \begin{tabular}{|l|l|l|}
     \hline
     \multicolumn{1}{|c|}{$?p$} &
     \multicolumn{1}{|c|}{$?z$} & 
     \multicolumn{1}{|c|}{$?s$} \\
     \hline
       Alice & $?z_1$ & Charlie \\
       Alice & $?z_2$ & David \\
       Bob  & $?z_3$ & Eric \\ 
   \hline
  \end{tabular}
  }
   $$
where $?z_1$, $?z_2$ and $?z_3$ are 3 fresh variables and:
\begin{center}
  $ G'_\ex = G_\ex \cup \{$
  (Alice, teaches, $?z_1$),
  (Charlie, studies, $?z_1$),
  (Alice, teaches, $?z_2$),
  (David, studies, $?z_2$),
  (Bob, teaches, $?z_3$), 
  (Eric, studies, $?z_3$) $\}$
\end{center}

\end{example}

\subsection{Queries}
\label{ssec:gral-query}

We consider three kinds of queries~:
$\clause{CONSTRUCT}$ queries, $\clause{SELECT}$ queries and
$\clause{\conselect}$ queries. We define
the semantics of queries from the semantics of patterns.  According to
Definition~\ref{def:gral-sem-pattern}, all patterns have a
graph-to-set-of-matches semantics.  In contrast, $\clause{CONSTRUCT}$
queries have a graph-to-graph semantics and $\clause{SELECT}$ queries
have a graph-to-multiset-of-solutions or graph-to-table semantics while
$\clause{\conselect}$ have a graph-to-graph-and-table semantics.

\begin{definition}[syntax of queries]
  \label{def:gral-syn-query}
  Let $S$ be a set of variables, $R$ a graph and $P$ a pattern. 
A \emph{query} $Q$ has one of the following three shapes: 
\begin{enumerate}
\item      $\clause{ CONSTRUCT } R \clause{ WHERE } P$
\item $\clause{ SELECT } S \clause{ WHERE } P$
\item $\clause{ \conselect\ } S, R \clause{ WHERE } P$
\end{enumerate}
\end{definition}

\begin{definition}[result of $\clause{ CONSTRUCT }$ queries]
  \label{def:gral-sem-query-construct}
  Given a pattern $P_1$ and a graph $R$
  consider the query
  $Q= \clause{ CONSTRUCT } R \clause{ WHERE } P_1$ 
   and the pattern 
  $P = P_1 \clause{ BUILD } R$. 
 The \emph{result} of the query $Q$ over a graph $G$, denoted $\Result_C(Q,G)$,
is the subgraph of $\seg{P}{G}$ image of $R$ by the set of matches $\sem{P}{G}$.
\end{definition}

Thus, the result of a $\clause{ CONSTRUCT }$ query $Q$ over a graph
$G$ is the graph $\Result_C(Q,G)= \sem{P}{G}(R)$ built by ``gluing''
the graphs $m(R)$ for all matches $m$ in $\sem{P}{G}$, where $m(R)$ is
a copy of $R$ with each variable $x\in\V(R)-\V(P)$ replaced by a fresh
variable (which means, fresh for each $m$ and each $x$).

\begin{example}
\label{ex:gral-sem-query-construct}

Consider the query: 
\begin{center}
\begin{tabular}{lll} 
  $ Q_{C,\ex} $ & $=$ & 
  $\clause{ CONSTRUCT }R_\ex\clause{ WHERE }\clause{ BASIC }(L_\ex) $ \\
  & $=$ & $\clause{ CONSTRUCT }\{$
($?p$, \mbox{teaches}, $?z$), ($?s$, \mbox{studies}, $?z$) $ \}$ \\
  & & $\clause{ WHERE }\clause{ BASIC }(\{$ 
($?p$, \mbox{teaches}, $?t$), ($?s$, \mbox{studies}, $?t$) $ \}$) \\
  \end{tabular}
\end{center}
The corresponding pattern $P_\ex$ and the value $\ulp_\ex:R_\ex \To G'_\ex $
of $P_\ex$ over $G_\ex$ are as in Example~\ref{ex:gral-sem-pattern}.
It follows that the result of the query $Q_{C,\ex}$ over $G_\ex$
is the subgraph of $G'_\ex$ image of $R_\ex$ by $\ulp_\ex$:
\begin{center}
  $ \Result_C(Q_{C,\ex},G_\ex) = \{$
  (Alice, teaches, $?z_1$),
  (Charlie, studies, $?z_1$),
  (Alice, teaches, $?z_2$),
  (David, studies, $?z_2$),
  (Bob, teaches, $?z_3$), 
  (Eric, studies, $?z_3$) $ \}$.
\end{center}

\end{example}

\begin{remark} 
\label{rem:gral-sparql}
CONSTRUCT queries in $\sparql$ are similar to CONSTRUCT queries
considered in this paper:
the variables in $\V(R)-\V(P_1)$ 
play the same role as the blank nodes in $\sparql$.
By considering BUILD patterns,
thanks to the functional orientation of the definition of patterns, 
our language allows BUILD subpatterns:
this is new and specific to the present study.
\end{remark}

For $\clause{ SELECT }$ queries we proceed as for
$\clause{ CONSTRUCT } $ queries: we define a transformation from each
$\clause{ SELECT }$ query $Q$ to a $\clause{ BUILD }$ pattern $P$ and
a transformation from the result of pattern $P$ to the result of query
$Q$.  Definition~\ref{def:gral-sem-query-select} below would deserve
more explanations. However this is not the subject of this paper, see
\cite{DEP2021} for details about how turning a table to a graph.

\begin{definition}[result of $\clause{SELECT}$ queries]
  \label{def:gral-sem-query-select}
For every set of variables $S=\{s_1,...,s_n\}$, let
$ \Graph(S)$ denote the graph made of the triples $(r,c_j,s_j)$
for $j\in \{1,...,n\}$
where $r$ is a fresh variable and $c_j$ is a fresh constant string for each $j$.
Given a pattern $P_1$ and a set of variables $S=\{s_1,...,s_n\}$ 
consider the query $Q= \clause{ SELECT } S \clause{ WHERE } P_1$ 
and the pattern $P = P_1 \clause{ BUILD } \Graph(S) $.
The value of $P$ over a graph $G$ is a set of matches $\sem{P}{G}$
which assignment table has $n+1$ columns, corresponding to the
variables $r,s_1,...,s_n$. 
The \emph{result} of the query $Q$ over a graph $G$, denoted $\Result_S(Q,G)$,
is the multiset of solutions made of the rows of the assignment table
of $\sem{P}{G}$ after dropping the column $r$.
\end{definition}

\begin{example}
  \label{ex:gral-sem-query-select}

Consider the query: 
\begin{center}
\begin{tabular}{lll} 
  $ Q_{S,\ex} $ & $=$ & 
  $\clause{ SELECT } \{?p,\,?s\} \clause{ WHERE }\clause{ BASIC }(L_\ex) $ \\
  \end{tabular}
\end{center}
Let $R_{S,\ex} = \Graph(\{?p,\,?s\}) = \{(?r,A_p,?p),\,(?r,A_s,?s)\}$
where $?r$ is a fresh variable and $A_p$, $A_s$ are fresh distinct strings.
Then the pattern corresponding to $Q_{S,\ex}$ is: 
  $$ P_{S,\ex} = 
  \clause{ BASIC }(L_\ex) \clause{ BUILD }R_{S,\ex}$$
The value of $ P_{S,\ex}$ over $G_\ex$ is:
$$ \ulp_{S,\ex}:R_{S,\ex} \To G'_{S,\ex} \;\mbox{ with }\;
   \tbl(\ulp_{S,\ex}) = 
\mbox { \begin{tabular}{|l|l|l|}
     \hline
     \multicolumn{1}{|c|}{$?r$} &
     \multicolumn{1}{|c|}{$?p$} & 
     \multicolumn{1}{|c|}{$?s$} \\
     \hline
       $?r_1$ & Alice & Charlie \\
       $?r_2$ & Alice & David \\
       $?r_3$ & Bob  & Eric \\ 
   \hline
  \end{tabular}
  }
   $$
where $?r_1$, $?r_2$ and $?r_3$ are 3 fresh variables and:
\begin{center}
  $ G'_{S,\ex} = G_\ex \cup \{$
  ($?r_1$, $A_p$, Alice), ($?r_1$, $A_s$, Charlie),
  ($?r_2$, $A_p$, Alice), \\ ($?r_2$, $A_s$, David), 
  ($?r_3$, $A_p$, Bob), ($?r_3$, $A_s$, Eric) $ \}$
\end{center}
It follows that:
\begin{center}
  $ \Result_S(Q_{S,\ex},G_\ex) = $
\begin{tabular}{|l|l|}
     \hline
     \multicolumn{1}{|c|}{$?p$} & 
     \multicolumn{1}{|c|}{$?s$} \\
     \hline
       Alice & Charlie \\
       Alice & David \\
       Bob  & Eric \\ 
   \hline
  \end{tabular} 
\end{center}

\end{example}

\begin{definition}[result of $\clause{ \conselect }$ queries]
  \label{def:gral-sem-query-conselect}
  Given a pattern $P_1$ a graph $R$ and a set of variables
  $S=\{s_1,...,s_n\}$. Let   $\Graph(S)$ be the graph as described in Definition~\ref{def:gral-sem-query-select}.
  consider the query:
   \hsp $Q= \clause{ \conselect }\   S, R \clause{ WHERE } P_1$ 
   and the pattern:
   \hsp $P = P_1 \clause{ BUILD } (\Graph(S) \cup R)$. 
 The \emph{result} of the query $Q$ over a graph $G$, denoted $\Result_{CS}(Q,G)$,
is the pair consisting of the subgraph of $\seg{P}{G}$ image of $R$ by the set of matches
$\sem{P}{G}$ and the multiset of solutions made of the rows of the assignment table
of $\sem{P}{G} (\Graph(S))$ after dropping the column $r$.
\end{definition}

\begin{example}
\label{ex:conselect}
   
   We illustrate here the \conselect\ queries through a toy
   example. The idea is to have a query that both returns a graph and
   a table as result. Typically it may be helpful when one wants to
   query statistical facts about the generated graph. Let us consider
   the database defined in Example~\ref{ex:database}. We propose to ask
   the following query
   which generates a graph representing professors and their
   supervised students accompanied with simple statistics about the
   number of students supervised by each professor.

$$
\begin{array}{lcl}
Q_{CS,\ex} & = & \clause{\conselect } S_{CS,\ex},P_{CS,\ex} \clause{ WHERE } P'_{CS,\ex}\\
S_{CS,\ex} & = & \{(?p, ?nbstudents)\} \\
P_{CS,\ex} & = & \{(?s,supervisedby,?p)\} \\
P'_{CS,\ex} & = & \clause{BASIC }(\{(?p, is, Professor), (?p, teaches, ?c),\\
          &   & \phantom{\clause{BASIC }(\{}(?s, is, Student),(?s, studies, ?c)\}) \\
&   &   \clause{BIND } (\clause{COUNT } (?s \clause{ BY } ?p)\,)
\clause{ AS } ?nbstudents
\end{array}
$$

The result $\Result_{CS}(Q,G)$ of this query is the list of professors
with the number of students they supervise (in our toy database, Alice
has two students, and Bob has one student) together with the graph of
students supervised by a professor. The expected graph and table are
displayed below:

$$  \begin{array}{cc}
          \begin{array}{l}
          \{(David, supervisedby, Alice),\\ 
            \phantom{\{} (Charlie, supervisedby, Alice), \\
            \phantom{\{} (Eric, supervisedby, Bob) \}
          \end{array}
          & 
          \begin{array}{|c|c|}
            \hline
            ?p    & ?nbstudents \\ \hline
            Alice & 2 \\ \hline
            Bob   & 1 \\ \hline
          \end{array}
    \end{array}$$
  \end{example}

\section{A Sound and Complete Calculus}
\label{sec:opsemantics}

In this section we propose a calculus for \emph{solving} patterns and
queries based on a relation over patterns called \emph{\rel}. It
computes values (i.e., sets of solutions)
of patterns 
(Definition~\ref{def:gral-sem-pattern}) and  results of queries
(Definitions~\ref{def:gral-sem-query-construct}, \ref{def:gral-sem-query-select} and~\ref{def:gral-sem-query-conselect}) over any graph.  This calculus is
sound and complete with respect to the set-theoretic semantics given
in Section~\ref{sec:gral}.

In functional and logic programming languages, narrowing~\cite{AntoyEH00} or
resolution~\cite{Lloyd87} derivations are used to solve goals and may have
the following shape where $g_0$ is the initial goal to solve (e.g.,
conjunction of atoms, equations  or a (boolean) term) 
and $g_{n+1}$ is a ``terminal'' goal such as the empty clause,
unifiable equations or the constant \emph{true} :
$$g_0 \rew_{[\sigma_0]} g_1  \rew_{[\sigma_1]} g_2 \ldots g_n  \rew_{[\sigma_n]} g_{n+1}$$
From such a derivation, a solution is obtained by simple composition
of local substitutions $\sigma_n \circ \ldots \sigma_1 \circ \sigma_0$
with restriction to variables of the initial goal $g_0$. In this
paper, $g_0$ is a pattern or a query and the underlying program is not
a set of clauses or rewriting rules but a graph augmented by a set of
rewriting rules defining the behavior of 
two functions $\Solve$ (for patterns) and $\Solve_Q$ (for queries).
An important difference between the setting developed in this
paper and classical functional and logic languages comes from the use
of functional composition ``$\circ$'' in
$\sigma_n \circ \ldots \sigma_1 \circ \sigma_0$.
Depending on the shape of the considered patterns, solutions can be obtained
by using additional composition operators such as $\opn{Join}$ 
(Definitions~\ref{def:algebra-opn} and \ref{def:gral-sem-pattern})
which composes only compatible substitutions computed by different
parts of a derivation (e.g.,
$\opn{Join}(\sigma_k \circ \ldots \sigma_0, \sigma_n \circ \ldots
\sigma_{k+1}$)\,).
In order to have an easy way to handle such kinds of
compositions when developing derivations starting 
from patterns, we introduce below the notion of configuration.
We write $\Pat$ for the sort of patterns.

\begin{definition}[configuration]
  Let $[\_, \_]: \Pat, \Som \to
  \configurations$ be the unique constructor operator of the sort
  $\configurations$. Let $\ulm: L \To G$ be a set of matches
  from graph $L$ to graph $G$ and $P$ a pattern. A
  configuration is denoted using a mixfix notation as a pair 
  $\configdeux{\ulm : L \To G}{P}$
  or simply $\configdeux{\ulm}{P}$.
  An \emph{initial} configuration is a configuration of the form
  $\configdeux{\uli_G: \emptygraph \To G}{P}$
  where $\uli_G =\opn{Match}(\emptygraph,G)$ is the set with one unique element
  that is the inclusion of the empty graph into $G$.
  A \emph{terminal} configuration is a configuration of the form
  $\configdeux{\ulm: L \To G}{\emptypattern}$.
\end{definition}

Roughly speaking, a configuration $\configdeux{\ulm : L \To G}{P}$
represents a state where the considered pattern is $P$, the current
graph database is $G$ which is the target of the current set of
matches $\ulm : L \To G$. 
Finding solutions of a pattern $P$ over a graph $G$
consists in starting from the term
$\Solve (\configdeux{\uli_G: \emptygraph \To G}{P})$
which applies the function $\Solve$ to an initial configuration and
then using appropriate rewriting rules to transform configurations
until reaching a terminal configuration of the form
$\configdeux{\ulm : L \To G'}{\emptypattern}$ where $\ulm : L \To G'$ represents the
expected set of matches (solutions) of $P$ over $G$ and where $G'$ is the
graph obtained after solving the pattern $P$ over $G$. Notice that
graph $G'$ contains $G$ but is not necessarily equal to $G$. 

In Fig.~\ref{fig:rew-patterns}, we provide a rewriting system,
$\syst$, which defines the function $\Solve$. This function is defined
by structural induction on the first component of configurations,
i.e., on the patterns. The second argument of configurations, i.e.,
the sets of matches, in the left-hand sides defining the function
$\Solve$ are always variables of the form $\ulm: L \To G$ or simply
$\ulm$ and thus can be handled easily in the pattern-matching process
of the left-hand sides of the proposed rules (no need to higher-order
pattern-matching nor unification).  In the rules of $\syst$, the
letters $P$, $P_1$ and $P_2$ are variables ranging over
\emph{patterns} (sort \Pat) while variables $L, G$ and $R$ are ranging
over \emph{graphs} (sort \Gr) and $\emptygraph$ is the constant
denoting the empty graph.
Symbol $e$ is a variable of sort $\Exp$ and $x$ is a variable of
subsort $\Var$ while $\ulm$, $\ulm'$ and $\ulp$ are variables of sort
$\Som$. Some constraints of the rules use operations already
introduced in Definition~\ref{def:algebra-opn}, such as $\opn{Match}$,
$\opn{Join}$, $\opn{Bind}$, $\opn{Filter}$, $\opn{Build}$ and
$\opn{Union}$.

\begin{figure}[!th]
  \caption{$\syst$: Rewriting rules for patterns}
  \label{fig:rew-patterns}
\resizebox{\textwidth}{!}{$\displaystyle
 \begin{array}{|llrcl|}
\hline
  &&&& \\
  \;r_0 \;\; & : &
  \Solve\, (\,\conf{\emptypattern}{\ulm:L\To G}\,) &
  \to & 
  \conf{\emptypattern}{\emptysetbis_G:\emptygraph\To G} \\ 
  &&&& \\
  \;r_1 \;\; & : &
  \Solve\, (\,\conf{\clause{BASIC}(L)}{\ulm:L\To G}\,) &
  \to & 
  \conf{\emptypattern}{\ulp:L\To G} \\ 
  &&&& \textrm{ where } \ulp=\opn{Match}(L,G) \\
  &&&& \\
  \;r_2 & : &
  \Solve\, (\,\conf{P_1 \clause{ JOIN } P_2}{\ulm}\,) &
  \to & 
  \Solve_{\it JL}\,(\Solve\, (\,\conf{P_1}{\ulm}\,), P_2) \\
  \;r_3 & : &
  \Solve_{\it JL}\, (\,\conf{\emptypattern}{\ulm},P) &
  \to & 
  \Solve_{\it JR}\,(\ulm, \Solve\,(\,\conf{P}{\ulm}\,) \\ 
  \;r_4 & : &
  \Solve_{\it JR}\,(\ulm, \conf{\emptypattern}{\ulm'}\, ) & 
  \to &
  \conf{\emptypattern}{\ulp} \\
  &&&& \textrm{ where } \ulp =  \opn{Join}(\ulm,\ulm') \\
  &&&& \\
  \;r_5 & : &
  \Solve\,(\,\conf{P \clause{ BIND }\expr \; \clause{AS} \;x}{\ulm}\,) &
  \to & 
  \Solve_{\it BI}\,(\Solve\,(\,\conf{P}{\ulm}\,), \expr, x) \\
  \;r_6 & : &
  \Solve_{\it BI}\,(\,\conf{\emptypattern}{\ulm}, \expr, x) &
  \to & 
  \conf{\emptypattern}{\ulp} \\
  &&&& \textrm{ where } \ulp = \opn{Bind}(\ulm,\expr,x) \\
  &&&& \\
  \;r_7 & : &
  \Solve\,(\,\conf{P \clause{ FILTER }\expr}{\ulm}\,) &
  \to & 
  \Solve_{\it FR}\,(\Solve\,(\,\conf{P}{\ulm}\,), \expr) \\
  \;r_8 & : &
  \Solve_{\it FR}\,(\,\conf{\emptypattern}{\ulm}, \expr) &
  \to & 
  \conf{\emptypattern}{\ulp} \\
  &&&& \textrm{ where } \ulp = \opn{Filter}(\ulm,\expr) \\
  &&&& \\
  \;r_9 & : &
  \Solve\,(\,\conf{P \clause{ BUILD }R}{\ulm}\,) &
  \to & 
  \Solve_{\it BU}(\Solve\,(\,\conf{P}{\ulm}\,), R) \\
  \;r_{10} & : &
  \Solve_{\it BU}\,(\,\conf{\emptypattern}{\ulm}, R) &
  \to & 
  \conf{\emptypattern}{\ulp} \\
  &&&& \textrm{ where } \ulp = \opn{Build}(\ulm,R) \\
  &&&& \\
  \;r_{11} & : &
  \Solve\,(\,\conf{P_1\clause{ UNION }P_2}{\ulm}\,) &
  \to & 
  \Solve_{\it UL}\,(\Solve\,(\,\conf{P_1}{\ulm}\,), P_2) \; \\
  \;r_{12} & : &
  \Solve_{\it UL}\,(\,\conf{\emptypattern}{\ulm}, P) &
  \to & 
  \Solve_{\it UR}\,(\ulm, \Solve\,(\,\conf{P}{\ulm}\,) \\
  \;r_{13} & : &
  \Solve_{\it UR}\,(\ulm,\,\conf{\emptypattern}{\ulm'}) &
  \to & 
  \conf{\emptypattern}{\ulp} \\
  &&&& \textrm{ where } \ulp = \opn{Union}(\ulm,\ulm') \\
  &&&& \\
\hline
 \end{array} $
 }
\end{figure}

\noindent
  In the sequel, we write $\patalgebra (\V)$ for the term algebra
  over the set of variables $\V$ generated by the operations occurring in
  the rewriting system $\syst$.

  Rule $r_0$ considers the degenerated case when one looks for
  solutions of the empty pattern $\emptypattern$. In this case there
  is no solution and the empty set of matches $\emptysetbis_G$ is computed.

  Rule $r_1$ is key in this calculus because it considers
  basic patterns of the form $\clause{BASIC}(L)$ where $L$ is a graph which may
  contain variables. In this case
  $\Solve\, (\,\conf{\clause{BASIC}(L)}{\ulm}\,)$ consists in finding all
  matches from $L$ to $G$. These matches can instantiate variables in $L$.
  Thus, the constraint $\ulp = \opn{Match}(L,G)$ of rule $r_1$ 
  instantiates variables occurring in graph $L$. This
  variable instantiation process is close to the narrowing or the resolution-based
  calculi.

  As said earlier the term
  $\Solve(\configdeux{\uli_G:\emptygraph \To G}{P})$ is intended to find
  both the solutions of pattern $P$ over graph $G$ and the 
  graph $G'$ obtained after transforming
  graph $G$ along the evaluation of the subpatterns of $P$. So,
  the aim of the \rel\ process is to 
  infer all solutions (matches) of a pattern $P$ over a graph
  $G$ starting from the term
  $\Solve(\configdeux{\uli_G:\emptygraph \To G}{P})$. 

In the context of functional-logic programming languages, several
strategies of narrowing-based procedures have been developed to solve
goals including even a needed strategy \cite{AntoyEH00}. In this paper, we
do not need all the power of narrowing procedures because manipulated
data are mostly flat (mainly constants and variables). Thus the
unification process used at every step in the narrowing relation is
beyond our needs. On the other hand, the classical rewriting relation
induced by the above rewriting system is not enough since variables in
patterns $P$ have to be instantiated and such an instantiation cannot
be done by simply rewriting the initial term
$\Solve(\configdeux{\uli_G: \emptygraph \To G}{P})$.

  Consequently, we propose hereafter a new relation
  induced by the above rewriting system that we call \rel. Before the
  definition of this relation, we recall briefly some notations about
  first-order terms. Readers not familiar with such notations may
  consult, e.g.,  \cite{TRSandAllThat}.

\begin{definition}[position, subterm replacement, substitution, $t\!\!\downarrow_{gq}$]
  A \emph{position} is a sequence of positive integers identifying a
  subterm in a term. For a term $t$, the empty sequence, denoted
  $\Lambda$, identifies $t$ itself. When $t$ is of the form
  $g(t_1,\ldots,t_n)$, the position $i.p$ of $t$ with
  $1 \leq i \leq n$ and $p$ is a position in $t_i$, identifies the
  subterm of $t_i$ at position $p$. The subterm of $t$ at position $p$
  is denoted $t|p$ and the result of replacing the subterm of $t$ at
  position $p$ with term $s$ is written $t[s]_p$.  We write
  $t\!\!\downarrow_{gq}$ for the term obtained from $t$ where all
  expressions of $\gqa$-algebra (i.e., operations such as
  $\opn{Join}$, $\opn{Bind}$, $\opn{Filter}$, $\opn{Match}$, etc.) have been evaluated.  A substitution $\sigma$ is a mapping from
  variables to terms. When $\sigma(x) = u$ with $u \not= x$, we say
  that $x$ is in the domain of $\sigma$. We write $\sigma(t)$ to
  denote the extension of the application of $\sigma$ to a term $t$
  which is defined inductively as $\sigma(c) = c$ if $c$ is a constant
  or $c$ is a variable outside the domain of $\sigma$. Otherwise
  $\sigma(f(t_1, \ldots, t_n)) = f(\sigma(t_1), \ldots,
  \sigma(t_n))$.
  \end{definition}

\begin{definition}[\rel\  $\rew$]
  \label{def:rew}
  The rewriting system $\syst$ defines a binary relation $\rew$
  over terms in $\patalgebra (\V)$ that we call \emph{\rel}\
  relation.We write $t \rew_{[u,lhs \to rhs, \sigma]} t'$ or simply
  $t \rew t'$ and say that $t$ is \emph{gq-narrowable}  to $t'$
  iff there exists a rule $lhs \to rhs$ in the rewriting system
  $\syst$, a position $u$ in $t$ and a substitution $\sigma$ such that
  $\sigma(lhs) = t|_u$ and
  $t' = t[\sigma(rhs)\!\!\downarrow_{gq}]_u$. Then $\rewstar$ denotes
  the reflexive and transitive closure of the relation $\rew$.
\end{definition}

Notice that in the definition of term
$t'= t[\sigma(rhs)\!\!\downarrow_{gq}]_u$ above, the substitution
$\sigma$ is not applied to $t$ as in narrowing
($\sigma(t[rhs]_u\downarrow_{gq}) $ but only to the right-hand side
($\sigma(rhs)$). This is mainly due to the possible use of additional
function composition such as $\opn{Join}$ operation.  If we consider
again the rule $r_1$, $t'$ would be of the following shape
$t' = t[\emptypattern,{(\opn{Match}(\sigma(L),G):\sigma(L)\To
  G)\!\!\downarrow_{gq}}]_u$. Notice that, in this case, the
evaluation of $\opn{Match}$ operation instantiates possible variables
occurring in the pattern $BASIC(\sigma(L))$ just like classical
narrowing procedures.

\begin{definition}[\rel\  derivations]
  Let $G$ be a graph, $P$ a pattern and 
  $\ulm$ a set of matches. The evaluation of $P$ over 
  $G$ consists in computing  \rel\ derivations of the form:
    $$ \Solve(\configdeux{\uli_G : \emptygraph \To G}{P}) \rewstar \configdeux{\ulm}{\emptypattern} $$
  \end{definition}

  \begin{example}
  \label{ex:opsem-pattern}
As in Example~\ref{ex:gral-sem-pattern} we consider the pattern:
\begin{center}
\begin{tabular}{lll} 
  $ P_\ex $ & $=$ & 
  $\clause{ BASIC }(L_\ex) \clause{ BUILD }R_\ex$ \\
  & $=$ & $\clause{ BASIC }(\{$ 
($?p$, \mbox{teaches}, $?t$), ($?s$, \mbox{studies}, $?t$) $ \}$) \\
  & & $\clause{ BUILD }\{$
($?p$, \mbox{teaches}, $?z$), ($?s$, \mbox{studies}, $?z$) $ \}$ \\
  \end{tabular}
\end{center}
The expected \rel\ derivation is as follows:
$$ \begin{array}{lll}
  \Solve\,(\,\conf{P_\ex}{\uli_{G_\ex}}\,) &
  \rew_{r_9} &
  \Solve_{\it BU}(\Solve\,(\,\conf{\clause{ BASIC }(L_\ex)}{\uli_{G_\ex}},R_\ex\,) \\
  & \rew_{r_1} &
  \Solve_{\it BU}(\,\conf{\emptypattern}{\Ulm(L_\ex,G_\ex},R_\ex\,) \\
  & \rew_{r_{10}} &
  \conf{\emptypattern}{\opn{Build}(\Ulm(L_\ex,G_\ex),R_\ex)} \\ 
\end{array} $$
According to Example~\ref{ex:gral-sem-pattern} this is the required result.
\end{example}

\begin{example}
  \label{ex:join-count}
  
We consider again Example~\ref{ex:database} and enrich the database
with  a few triples stating membership to a lab for professors and
fixing supervisors of some students.

\begin{center}
\begin{tabular}{llll} 
  $G_A= \{
  $ & (Alice, is, Professor), & (Alice, teaches, Mathematics), & \\
    & (Bob, is, Professor), & (Bob, teaches, Informatics), & \\
    & (Charlie, is, Student), & (Charlie, studies, Mathematics), & \\
    & (David, is, Student), & (David, studies, Mathematics), & \\ 
    & (Eric, is, Student), & (Eric, studies, Informatics), & \\
    & (Alice, member, Lab1), & (Bob, member, Lab2), & \\
    & (David, supervisedby, Alice), & (Eric, supervisedby, Bob) \ \ $ \}$& \\
  \end{tabular}
\end{center}

\noindent
We illustrate below a \rel\ derivation which solves a pattern $\pi_A$
over graph $G_A$. The goal from this pattern is to find students who
are interns in some laboratory (set of matches) and add them to the database
(new database). Let $\pi_A$ be the pattern $ \pi_1 \, \clause{JOIN} \, \pi_2 $ 
where:\\

\noindent
$\pi_1 = (P1  \,\ \clause{BUILD} \,\ R1)$,\\
$\pi_2 =  (P2  \,\ \clause{BUILD} \,\ R2)$,\\
$P1= BASIC(L_1)$, with
$L_1 = \{(?x, supervisedby, ?p), (?p, member , ?l)\})$, \\
$R1= \{(?x, member, ?l)\}$, \\
$P2= BASIC(L_2)$, with 
$L_2 = \{(?x, member, ?t), (?x, is, Student)\}$, and\\
$R2 = \{(?x, is, Intern)\}$.

Solutions for pattern $\pi_A$ over graph $G_A$ is obtained by
performing the following \rel\ derivation in which we omit to specify
the domains and codomains of the involved sets of matches. They are 
defined after the derivation. The evaluation of the subpattern $\pi_1$
transforms graph $G_A$ by adding two additional triples $\{ (David,
member, Lab1), (Eric, member, Lab2)\}$. The evaluation of subpattern
$\pi_2$ uses these triples and adds itself two new triples \\ $\{(David,
is, Intern), (Eric, is, Intern)\}$.

{\small
$$
\begin{array}{l}
Solve(\configdeux{\uli_{G_A}: \emptygraph \To G_A}{\pi_A}) \\
  \begin{array}{lll}
    (1)    & \rew_{r_2} & Solve_{JL}(Solve(\configdeux{\uli_{G_A}}{\pi_1}), \pi_2)\\
    (2)~~~ & \rew_{r_9} &
                          Solve_{JL}(Solve_{BU}(Solve(\configdeux{\uli_{G_A}}{P1}), R1), \pi_2)\\
    (3)    & \rew_{r_1}    & Solve_{JL}(Solve_{BU}( \conf{\emptypattern}{\ulp_1: L_1 \To G_A}, R1), \pi_2)\\
    (4)    & \rew_{r_{10}} & Solve_{JL}( \conf{\emptypattern}{\ulp_2}, \pi_2)\\
    (5)    & \rew_{r3} & Solve_{JR}( \ulp_2, Solve(\configdeux{\ulp_2}{\pi_2}))\\
    (6)    & \rew_{r9} & Solve_{JR}( \ulp_2, Solve_{BU}(Solve(\configdeux{\ulp_2}{P_2}),R2))\\ 
    (7)    & \rew_{r6} & Solve_{JR}( \ulp_2, Solve_{BU}(\configdeux{\ulp_3}{\emptypattern},R2))\\
    (8)    & \rew_{r_{10}} &  Solve_{JR}( \ulp_2, \configdeux{\ulp_4}{\emptypattern})\\
    (9)    & \rew_{4} & \configdeux{\ulp_5}{\emptypattern}
    \end{array}
\end{array}
$$}

\noindent
Where: \\
\noindent
$\ulp_1 = \Ulm(L_1, G_A)$\\
$\ulp_2 =  \opn{Build}(\ulp_1,R1): R1 \To G_B$ with $G_B = G_A \cup
H_{\ulp_1, R1}$\\
$\ulp_3 =  \Ulm(L_2,G_1): L_2 \To G_B$ \\
$\ulp_4 =  \opn{Build}(\ulp_3,R2): R2 \To G_C$ with $G_C = G_B \cup H_{\ulp_3,R2}$\\
$\ulp_5 =  \opn{Join}(\ulp_2, \ulp_4) : R1 \cup R2 \To G_B \cup G_C$
\\

There are two matches from $L_1$ to $G_A$. 
The set $\ulp_1$ of these matches is: 
 $$ \ulp_1 :L_1 \To G_A \;\mbox{ with }\;
   \tbl(\ulp_1) = 
\mbox { \begin{tabular}{|l|l|l|}
     \hline
     \multicolumn{1}{|c|}{$?x$} &
     \multicolumn{1}{|c|}{$?p$} & 
     \multicolumn{1}{|c|}{$?l$} \\
     \hline
       David & Alice & Lab1 \\
       Eric & Bob & Lab2 \\
   \hline
  \end{tabular}
  }
   $$

Let $G_B$ be the intermediate graph, $G_B = G_A \cup
H_{\ulp_1, R1}$, as depicted below.

\begin{center}
\begin{tabular}{llll} 
  $G_{B}= \{
  $ & (Alice, is, Professor), & (Alice, teaches, Mathematics), & \\
    & (Bob, is, Professor), & (Bob, teaches, Informatics), & \\
    & (Charlie, is, Student), & (Charlie, studies, Mathematics), & \\
    & (David, is, Student), & (David, studies, Mathematics), & \\ 
    & (Eric, is, Student), & (Eric, studies, Informatics), & \\
    & (Alice, member, Lab1), & (Bob, member, Lab2), & \\
    & (David, supervisedby, Alice), & (Eric, supervisedby, Bob), & \\
    & (David, member, Lab1), & (Eric, member, Lab2), $\}$ \\
  \end{tabular} 
\end{center}

Then, the set of matches $\ulp_2 : R1 \To G_B$ consists of two
matches.

 $$ \ulp_2 :R1 \To G_B \;\mbox{ with }\;
   \tbl(\ulp_2) = 
\mbox { \begin{tabular}{|l|l|l|}
     \hline
     \multicolumn{1}{|c|}{$?x$} &
     \multicolumn{1}{|c|}{$?l$} \\ 
     \hline
       David & Lab1  \\
       Eric & Lab2  \\
   \hline
  \end{tabular}
  }
   $$

There are two matches from $L_2$ to $G_B$. 
The set $\ulp_3$ of these matches is: 
 $$ \ulp_3 :L_2 \To G_B \;\mbox{ with }\;
   \tbl(\ulp_3) = 
\mbox { \begin{tabular}{|l|l|l|}
     \hline
     \multicolumn{1}{|c|}{$?w$} &
     \multicolumn{1}{|c|}{$?t$} \\
     \hline
       David & Lab1  \\
       Eric & Lab2 \\
   \hline
  \end{tabular}
  }
   $$

Let $G_C$ be the following graph   $G_C = G_B \cup H_{\ulp_3,R2}$:  

\begin{center}
\begin{tabular}{llll} 
  $G_{C}= \{
  $ & (Alice, is, Professor), & (Alice, teaches, Mathematics), & \\
    & (Bob, is, Professor), & (Bob, teaches, Informatics), & \\
    & (Charlie, is, Student), & (Charlie, studies, Mathematics), & \\
    & (David, is, Student), & (David, studies, Mathematics), & \\ 
    & (Eric, is, Student), & (Eric, studies, Informatics), & \\
    & (Alice, member, Lab1), & (Bob, member, Lab2), & \\
    & (David, supervisedby, Alice), & (Eric, supervisedby, Bob), & \\
    & (David, member, Lab1), & (Eric, member, Lab2), \\
    & (David, is, Intern), & (Eric, is, Intern) \ \ $ \}$& \\
  \end{tabular}
\end{center}

The set of matches $\ulp_4: R2 \To G_C$  consists of two
matches.

 $$ \ulp_4 :R2 \To G_C \;\mbox{ with }\;
   \tbl(\ulp_4) = 
\mbox { \begin{tabular}{|l|l|l|}
     \hline
         \multicolumn{1}{|c|}{$?w$} \\ 
     \hline
       David \\
       Eric  \\
   \hline
  \end{tabular}
  }
   $$

   Finally, $\ulp_5 =  \opn{Join}(\ulp_2, \ulp_4) : R1 \cup R2 \To
   G_B \cup G_C$ consists of two matches

    $$ \ulp_5 : R1 \cup R2 \To
   G_B \cup G_C \;\mbox{ with }\;
   \tbl(\ulp_5) = 
\mbox { \begin{tabular}{|l|l|l|}
     \hline
     \multicolumn{1}{|c|}{$?x$} &
     \multicolumn{1}{|c|}{$?l$} \\
     \hline
          David & Lab1   \\
          Eric & Lab2   \\
   \hline
  \end{tabular}
  }
   $$

\end{example}

\begin{theorem}[soundness]
  \label{th:soundness}
  Let $G$ be a graph, $P$ a pattern and $\ulm$ a set of matches such that
  $ \Solve(\configdeux{\uli_G}{P})
  \rewstar \configdeux{\ulm}{\emptypattern} $. Then for all morphisms
  $m$ in $\ulm$, there exists a morphism $m'$ equals to $m$ up to
  renaming of variables such that $m'$ is in $\sem{P}{G}$.
\end{theorem}

\begin{proof}
  The proof is done by induction on the length $n$ of derivation \\
   $ Solve(\config{\uli_{G}}{P}{G})
   \rew^{n}(\config{\ulm}{\emptypattern}{G'}) $.

   \emph{Base case}. $n=1$. In this case only rules $r_0$ or $r_1$ are
   possible.
   \begin{description}
     \item Case of rule $r_0$. In this case, the pattern $P$ is the
       empty pattern $\emptypattern$ and $\ulm = \emptysetbis_G$. The case
       vacuously holds.
     \item Case of rule $r_1$. In this case, the pattern $P$ is  of
       the form $\clause{BASIC}(L)$ and the considered derivation is of the form
    $ Solve (\config{\uli_{G}}{\clause{BASIC}(L)}{G}) \rew
             \config{\ulm : L \To G} {\emptypattern} {G_2}$ where
             $\ulm =  \Ulm(L,G)$. The claim obviously holds since
             $\sem{\clause{BASIC}(L)}{G} = \Ulm(L,G):L\To G$ by
             Definition~\ref{def:gral-sem-pattern}.
                   \end{description}

     \emph{Induction Step}. $n > 1$. In this case the pattern $P$ can
     be of five different shapes as discussed below:
     \begin{description}
     \item [$P = \pi_1 \clause{ JOIN } \pi_2$.] Then, the \rel\ 
       derivation  $ Solve(\config{\uli_{G}}{P}{G}) $ $
   \rew^{n}(\config{\ulm}{\emptypattern}{G'}) $ has the following
   shape

   $ Solve(\config{\uli_{G}}{\pi_1 \clause{ JOIN } \pi_2}{G}) $ $ \rew 
   Solve_{JL}(Solve(\config{\uli_{G}}{\pi_1}{G}), \pi_2) $ $
   \rew^{n_{1}} \\
   Solve_{JL}(\config{\ulm_1}{\emptypattern}{G_1}), \pi_2) $ $ \rew
   Solve_{JR}(\ulm_1, Solve(\config{\ulm_1}{\pi_2}{G_1})) $ $
   \rew^{n_{2}} \\
   Solve_{JR}(\ulm_1, \config{\ulm_2: Z \To G'}{\emptypattern}{G'}) $
   $ \rew
   \config{\ulm_3}{\emptypattern}{G'}$ with
   $\ulm_3 = \opn{Join}(\ulm_1,\ulm_2)$ and $Z$ a graph.

   Notice that the length $n$ equals $n_1+n_2+3$. By induction
   hypothesis, the set of matches $\ulm_1$ and $\ulm_2$ are
   sound. Then, the set $\ulm_3$ is obtained by
   using the operation $\opn{Join}$ over the sets $\ulm_1$ and
   $\ulm_2$ which ensures the soundness of the set $\ulm_3$.

 \item [$P = \pi_1\clause{ BIND }\expr_1\clause{ AS } x_1$.] Then, the \rel
\    derivation  $ Solve(\config{\uli_{G}}{P}{G}) $ $
   \rew^{n}(\config{\ulm}{\emptypattern}{G'}) $ has the following
   shape:

   $ Solve(\config{\uli_{G}}{\pi_1\clause{ BIND }\expr_1\clause{ AS }
     x_1}{G}) $ $\rew 
   Solve_{BD}(Solve(\config{\uli_{G}}{\pi_1}{G}), \expr_1, x_1) $ $
   \rew^{n_{1}} \\
   Solve_{BD}(\config{\ulm_1: Z \To G_1}{\emptypattern}{G_1}),
   \expr_1, x_1) $ $ \rew
   \config{\ulm_2}{\emptypattern}{G'}$ with
    $G' =  G_1 \cup \ulm(\expr_1)$ and $\ulm_2 = \opn{Bind}(\ulm_1,\expr_1,x_1) : Z \cup \{x_1\} \To   G'$

    Notice that the length $n$ equals $n_1+2$. Thus $n_1$ is less than
    $n$ and by induction hypothesis, the set of matches $\ulm_1$ is
    sound. Then, the set $\ulm_2$ is obtained as
    $\opn{Bind}(\ulm_1,\ev(\ulm_1,\expr_1),x_1)$ as expected by the
    semantics which ensures the soundness of the set $\ulm_2$.

  \item [$P = \pi_1 \clause{ FILTER }\expr_1 $.]
    Then, the \rel\ derivation  $ Solve(\config{\uli_{G}}{P}{G}) $ $
   \rew^{n}(\config{\ulm}{\emptypattern}{G'}) $ has the following
   shape:

   $ Solve(\config{\uli_{G}}{\pi_1\clause{ FILTER }\expr_1}{G}) $ $ \rew 
   Solve_{FR}(Solve(\config{\uli_{G}}{\pi_1}{G}), \expr_1) $ $
   \rew^{n_{1}} \\
   Solve_{FR}(\config{\ulm_1: Z \To G'}{\emptypattern}{G'}),
   \expr_1) $ $ \rew
   \config{\ulm_2}{\emptypattern}{G'}$ with
    $\ulm_2 = \opn{Filter}(\ulm_1,\expr_1) : Z  \To   G'$

    Notice that the length $n$ equals $n_1+2$. Thus $n_1$ is less than
    $n$ and by induction hypothesis, the set of matches $\ulm_1$ is
    sound. Then, the set $\ulm_2$ is obtained as
   $\opn{Filter}(\ulm_1,\ev(\ulm_1,\expr_1))$ as expected by the
    semantics which ensures the soundness of the set $\ulm_2$.

  \item [$P = \pi_1 \clause{ BUILD }R $.]
    Then, the \rel\ derivation  $ Solve(\config{\uli_{G}}{P}{G}) $ $
   \rew^{n}(\config{\ulm}{\emptypattern}{G'}) $ has the following
   shape:

   $ Solve(\config{\uli_{G}}{\pi_1\clause{ BUILD } R_1}{G})$ $ \rew$ $
   Solve_{BU}(Solve(\config{\uli_{G}}{\pi_1}{G}), R_1) $
   $\rew^{n_{1}}$ $ \\
   Solve_{BU}(\config{\ulm_1: Z \To G_1}{\emptypattern}{G_1}), R_1)
   $ $\rew $ $\config{\ulm_2}{\emptypattern}{G'}$ with
   $G' = G_1 \cup \ulm_1^{\sharp}(R_1)$ and
   $\ulm_2 = \opn{Build}(\ulm_1,R_1) : Z \To G'$.

    Notice that the length $n$ equals $n_1+2$. Thus $n_1$ is less than
    $n$ and by  induction hypothesis, the set of matches $\ulm_1$ is
    sound. Then, the set $\ulm_2$ is obtained as
    $\opn{Build}(\ulm_1,R_1)$  as expected by the
    semantics which ensures the soundness of the set $\ulm_2$.

     \item [$P = \pi_1\clause{ UNION }\pi_2 $.] Then, the \rel\
       derivation  $ Solve(\config{\uli_{G}}{P}{G}) $ $
   \rew^{n}(\config{\ulm}{\emptypattern}{G'}) $ has the following
   shape

   $ Solve(\config{\uli_{G}}{\pi_1 \clause{ UNION } \pi_2}{G}) \rew 
   Solve_{UL}(Solve(\config{\uli_{G}}{\pi_1}{G}), \pi_2) $
   $\rew^{n_{1}} \\
   Solve_{UL}(\config{\ulm_1}{\emptypattern}{G_1}), \pi_2) $ $\rew
   Solve_{UR}(\ulm_1, Solve(\config{\ulm_1}{\pi_2}{G_1})) $ $
   \rew^{n_{2}} \\
   Solve_{UR}(\ulm_1, \config{\ulm_2: Z \To G'}{\emptypattern}{G'}) $ $ \rew
   \config{\ulm_3}{\emptypattern}{G'}$ with
   $\ulm_3 = \opn{Union}(\ulm_1,\ulm_2) : Z \To G'$ and $Z$ a graph.

   Notice that the length $n$ equals $n_1+n_2+3$. By induction
   hypothesis, the set of matches $\ulm_1$ and $\ulm_2$ are
   sound. Then, the set $\ulm_3$ is obtained by using the operation
   $\opn{Union}$ over the sets $\ulm_1$ and $\ulm_2$ as expected by
   the semantics which ensures the soundness of the set $\ulm_3$.
         
     \end{description}
  \end{proof}

\begin{theorem}[completeness]
  \label{th:completeness}
  Let $G_1, G_2$ and $X$ be graphs, $P$ a pattern and $h: X \To G_2$ a
  match in $\sem{P}{G_1}$. Then there exist graphs $G'_2$ and $X'$, a
  set of matches $\ulm : X' \To G'_2$, a derivation
  $ \Solve(\config{\uli_{G_1}}{P}{G_1} \rewstar
  \config{\ulm}{\emptypattern}{G'_2} $ and a match $m : X' \To G'_2$ in
  $\ulm$ such that $m$ and $h$ are equal up to variable renaming.
\end{theorem}

\begin{proof}
  The proof is done by structural induction over patterns.
  
  \emph{Base case}.
  $P = \emptypattern$. In this case, $\sem{P}{G_1}$ is empty and thus
  the statement vacuously holds. 

  \emph{Induction step}. There are six cases to consider according to
  the shape of pattern $P$.
  \begin{description}
     \item [$P = \clause{BASIC}(L)$.] In this case the set of
       matches $\sem{P}{G_1}$ coincides with the set $\ulm$ obtained
       after one step \rel\ 
       $ Solve (\config{\uli_{G_1}}{\clause{BASIC}(L)}{G_1}) \rew_{r_1}
             \config{\ulm : L \To G_1} {\emptypattern} {G}$ where
             $\ulm =  \Ulm(L,G_1)$. Obviously, $\ulm = \sem{P}{G}$.
       
           \item [$P = \pi_1 \clause{ JOIN } \pi_2$.] By
             Definition~\ref{def:gral-sem-pattern},
             $\sem{\pi_1\clause{ JOIN }\pi_2}{G_1} =
             \opn{Join}(\sem{\pi_1}{G_1},\sem{\pi_2}{\seg{\pi_1}{G_1}})$. Thus,
             $h$ being an element of $\sem{\pi_1\clause{ JOIN }\pi_2}{G_1}$
             there exist two matches $h_1$ in $\sem{\pi_1}{G_1}$ and
             $h_2$ in $\sem{\pi_2}{G_1^{(\pi_1)}}$ such that
             $h = h_1 \bowtie h_2$.  Let us consider the following
             \rel\ step
             $ Solve(\config{\uli_{G_1}}{\pi_1 \clause{ JOIN } \pi_2}{G_1}) \rew
             Solve_{JL}(Solve(\config{\uli_{G_1}}{\pi_1}{G_1}), \pi_2)$.  By
             induction hypothesis, there exists a derivation
             $Solve(\config{\uli_{G_1}}{\pi_1}{G_1}) \rewstar
             \config{\ulm_1}{\emptypattern}{G'_1}$ and a match
             $m_1$ in $\ulm_1$ such that $h_1$ and $m_1$ are equal up to
             variable renaming. Notice that $G'_1$ is isomorphic to
             $G^{({\pi_1})}$. Now we can develop further the above
             derivation and get
             $ Solve(\config{\uli_{G_1}}{\pi_1 \clause{ UNION } \pi_2}{G_1}) \rew
             Solve_{UL}(Solve(\config{\uli_{G_1}}{\pi_1}{G_1}), \pi_2)$ $ \rewstar
             Solve_{JL}(\config{\ulm_1}{\emptypattern}{G'_1}), \pi_2) \rew
             Solve_{JR}(\ulm_1, Solve(\config{\ulm_1}{\pi_2}{G'_1}))$.
             Again, by induction hypothesis, there exists a derivation
             $Solve(\config{\ulm_1}{\pi_2}{G'_1}) \rewstar
             \config{\ulm_2: Z \To G_2}{\emptypattern}{G_2}$ and a
             match $m_2$ in $\ulm_2$ such that $h_2$ and $m_2$ are equal up
             to variable renaming.

             Finally, we get the expected derivation

             $ Solve(\config{\uli_{G_1}}{\pi_1 \clause{ JOIN } \pi_2}{G_1})
             \rew Solve_{JL}(Solve(\config{\uli_{G_1}}{\pi_1}{G_1}),
             \pi_2)$
             $ \rewstar \\
             Solve_{JL}(\config{\ulm_1}{\emptypattern}{G'_1}), \pi_2)
             \rew Solve_{JR}(\ulm_1,
             Solve(\config{\ulm_1}{\pi_2}{G'_1})) \rewstar
             Solve_{JR}(\ulm_1, \config{\ulm_2: Z \To
               G_2}{\emptypattern}{G_2}) \rew
             \config{\ulm_3}{\emptypattern}{G_3}$ with
             $\ulm_3 = \opn{Join}(\ulm_1,\ulm_2)$. Therefore, by
             definition of $\ulm_3$, we have $m = m_1 \bowtie m_2$ is
             in $\ulm_3$ and $m$ is thus equal to $h$ up to variable
             renaming.

           \item [$P = \pi_1\clause{ BIND }\expr_1\clause{ AS } x_1$.]
             By Definition~\ref{def:gral-sem-pattern},
             $\sem{\pi_1\clause{ BIND }\expr_1\clause{ AS } x_1}{G_1}
             =
             \opn{Bind}(\sem{P_1}{G_1},\expr_1,x_1)
             : \se{P_1}\cup\{x_1\}\To
             \seg{P_1}{G_1}\cup{\sem{P_1}{G_1}(\expr_1)} $. Thus, $h$
             is an element of
             $
             \opn{Bind}(\sem{P_1}{G_1},\expr_1 ,x_1)$;
             Now, let us consider the following \rel\ step \\
             $ Solve(\config{\uli_{G_1}}{\pi_1\clause{ BIND
               }\expr_1\clause{ AS } x_1}{G_1}) \rew
             Solve_{BD}(Solve(\config{\uli_{G_1}}{\pi_1}{G_1}), \expr_1, x_1)$.
             By induction hypothesis, there exists a derivation
             $Solve(\config{\uli_{G_1}}{\pi_1}{G_1}) \rewstar
             \config{\ulm_1}{\emptypattern}{G'_1}$ such that for all
             matches $h_1$ in  $\sem{P_1}{G_1}$, there exists a match
             $m_1$ in $\ulm_1$ equal to $h_1$ up to variable
             renaming. Therefore, for every element $h$ in
             $\opn{Bind}(\sem{P_1}{G_1},\expr_1,x_1)$
              there exists a match $m$ in
              $\opn{Bind}(\ulm_1,\expr_1,x_1)$ with $m$
              and $h$ equal up to variable renaming.

            \item [$P = \pi_1 \clause{ FILTER }\expr_1 $.]
                           By Definition~\ref{def:gral-sem-pattern},
             $\sem{\pi_1\clause{ FILTER }\expr_1}{G_1}
             =
             \opn{Filter}(\sem{P_1}{G_1},\expr_1)
             : \se{P_1} \To
             \seg{P_1}{G_1} $. Thus, $h$
             is an element of
             $
             \opn{Filter}(\sem{P_1}{G_1},\expr_1)$;
             Now, let us consider the following \rel\ step
             $ Solve(\config{\uli_{G_1}}{\pi_1\clause{ FILTER
               }\expr_1}{G_1}) \rew
             Solve_{FR}(Solve(\config{\uli_{G_1}}{\pi_1}{G_1}), \expr_1)$.
             By induction hypothesis, there exists a derivation
             $Solve(\config{\uli_{G_1}}{\pi_1}{G_1}) \rewstar
             \config{\ulm_1}{\emptypattern}{G'_1}$ such that for all
             matches $h_1$ in  $\sem{P_1}{G_1}$, there exists a match
             $m_1$ in $\ulm_1$ equal to $h_1$ up to variable
             renaming. Therefore, for every element $h$ in
             $\opn{Filter}(\sem{P_1}{G_1},\expr_1)$
              there exists a match $m$ in
              $\opn{Filter}(\ulm_1,\expr_1)$ with $m$
              and $h$ equal up to variable renaming.

            \item [$P = \pi_1 \clause{ BUILD } R_1 $.]
                              By Definition~\ref{def:gral-sem-pattern},
             $\sem{\pi_1\clause{ BUILD } R_1}{G_1}
             =
             \opn{Build}(\sem{P_1}{G_1},R_1)
              : R_1\To \seg{P_1}{G_1} \cup \sem{P_1}{G_1}(R_1)$. Thus, $h$
             is an element of
             $
             \opn{Build}(\sem{P_1}{G_1}, R_1)$;
             Now, let us consider the following \rel\ step \\
             $ Solve(\config{\uli_{G_1}}{\pi_1\clause{ BUILD
               }R_1}{G_1}) \rew
             Solve_{BU}(Solve(\config{\uli_{G_1}}{\pi_1}{G_1}), R_1)$.
             By induction hypothesis, there exists a derivation
             $Solve(\config{\uli_{G_1}}{\pi_1}{G_1}) \rewstar
             \config{\ulm_1}{\emptypattern}{G'_1}$ such that for all
             matches $h_1$ in  $\sem{P_1}{G_1}$, there exists a match
             $m_1$ in $\ulm_1$ equal to $h_1$ up to variable
             renaming. Therefore, for every element $h$ in
             $\opn{Build}(\sem{P_1}{G_1}, R_1)$
              there exists a match $m$ in
              $\opn{Build}(\ulm_1, R_1)$ with $m$
              and $h$ equal up to variable renaming.

            \item [$P = \pi_1\clause{ UNION }\pi_2 $.]  By
              Definition~\ref{def:gral-sem-pattern},
              $\sem{\pi_1\clause{ UNION }\pi_2}{G_1} =
              \opn{Union}(\sem{\pi_1}{G_1},\sem{\pi_2}{\seg{\pi_1}{G_1}})$. Thus,
              $h$ being an element of
              $\sem{\pi_1\clause{ UNION }\pi_2}{G_1}$ either $h$ is an
              extension of a match $h_1$ in $\sem{\pi_1}{G_1}$ or a
              match $h_2$ in $\sem{\pi_2}{G_1^{(\pi_1)}}$.

             Let us consider the following \rel\ step
             $ Solve(\config{\uli_{G_1}}{\pi_1 \clause{ UNION }
               \pi_2}{G_1}) \rew
             Solve_{UL}(Solve(\config{\uli_{G_1}}{\pi_1}{G_1}), \pi_2)$.
             By induction hypothesis, there exists a derivation
             $Solve(\config{\uli_{G_1}}{\pi_1}{G_1}) \rewstar
             \config{\ulm_1}{\emptypattern}{G'_1}$.  If $h$ is an
             extension of $h_1$ element of $\sem{\pi_1}{G_1}$ then, by
             induction hypothesis, there exists a match $m_1$ in
             $\ulm_1$ such that $h_1$ and $m_1$ are equal up to
             variable renaming and $G'_1$ is isomorphic to
             $G^{({\pi_1})}$. Now we can develop further the above
             derivation and get
             $ Solve(\config{\uli_{G_1}}{\pi_1 \clause{ UNION } \pi_2}{G_1})
             \rew Solve_{UL}(Solve(\config{\uli_{G_1}}{\pi_1}{G_1}),
             \pi_2)$
             $ \rewstar
             Solve_{UL}(\config{\ulm_1}{\emptypattern}{G'_1}), \pi_2)
             \rew \\ Solve_{UR}(\ulm_1,
             Solve(\config{\ulm_1}{\pi_2}{G'_1}))$.  Again, by
             induction hypothesis, there exists a derivation
             $Solve(\config{\ulm_1}{\pi_2}{G'_1}) \rewstar
             \config{\ulm_2: Z \To G_2}{\emptypattern}{G_2}$ .
             If $h$ is an
             extension of $h_2$ element of $\sem{\pi_2}{G'_1}$ then, by
             induction hypothesis, there exists a match $m_2$ in
             $\ulm_2$ such that $h_2$ and $m_2$ are equal up to
             variable renaming and $G_2$ is isomorphic to
             ${G_1^{({\pi_1})}}^{(\pi_2)}$.
             Finally, we get the expected derivation

             $ Solve(\config{\uli_{G_1}}{\pi_1 \clause{ UNION } \pi_2}{G_1})
             \rew Solve_{UL}(Solve(\config{\uli_{G_1}}{\pi_1}{G_1}),
             \pi_2)$
             $ \rewstar \\
             Solve_{UL}(\config{\ulm_1}{\emptypattern}{G'_1}), \pi_2)
             \rew Solve_{UR}(\ulm_1,
             Solve(\config{\ulm_1}{\pi_2}{G'_1})) \rewstar
             Solve_{UR}(\ulm_1, \config{\ulm_2: Z \To
               G_2}{\emptypattern}{G_2}) \rew
             \config{\ulm_3}{\emptypattern}{G_3}$ with
             $\ulm_3 = \opn{Union}(\ulm_1,\ulm_2)$. Therefore, by
             definition of $\ulm_3$, we have $m$ can be either an
             extension of $m_1$ or $ m_2$ 
             in $\ulm_3$ and thus $m$ is equal to $h$ up to variable
             renaming.
   
  \end{description}
  \end{proof}

All \rel\ derivation steps for solving patterns are needed 
since at each step only one position is candidate to a \rel\ step.

\begin{proposition}[determinism] 
  \label{prop:deterministic}
  Let $ t_0 \rew t_1 \rew \ldots \rew t_n$ be a \rel\
  derivation with $t_0 = \Solve(\configdeux{\uli_G}{P})$.  For all
  $i \in [0..n]$, there exists at most one position $u_i$ in $t_i$
  such that $t_i$ can be gq-narrowed into $t_{i+1}$.
  \end{proposition}

  \begin{proof}
    The proof is by induction on $i$.

    Base case ($i = 0$). Since the function $solve$ is completely
    defined by structural induction on patterns, it follows that there
    exists one and only one rule that applies to the term 
    $t_0 = Solve(\config{\uli_{G}: \emptyset \To G}{P}{G})$ at
    position $\Lambda$ according to
    the structure (constructor at the head) of pattern $P$.

    Induction step: Assume that $t_i$ contains only one reducible
    subterm at position $u_i = u.k$ via rule $r$ with $k$ being a
    natural number ($k > 0$) and $u$ a position in $t_i$.  Below, we
    discuss the different cases according to the considered rule $r$.
    \begin{description}
    \item[$r = r_0$:]  Rule $r_0$ considers the degenerated case when
      one looks fo solutions of the empty pattern $\emptypattern$. In
      this case there is no solution and the empty set of matches is
      computed $\emptysetbis_G$. The right-hand side of rule $r_0$ is in normal
      form and thus does not contain a possible reducible
      subterm. In addition, according to the shape of the other rules only
      the next upper position $u$ may become a potential reducible
      term via one of the rules $r_0, r_3, r_4, r_6, r_8, r_{10},r_{12}, r_{13}$.
 
    \item[$r = r_1$:] Rule $r_1$ considers the case when one looks for
      solutions when the pattern is of the form $BASIC(G_1)$ for some
      graph $G_1$. The right-hand side of the rule is in normal form
      and thus does not contain a possible reducible subterm. However,
      according to the shape of the other rules (height of the
      left-hand sides is equal to 1), only the next upper term at position $u$
      may become a potential reducible term via one of the rules
      $r_0, r_3, r_4, r_6, r_8, r_{10}, r_{12}, r_{13}$.

    \item[$r= r_2$:] In this case we have $t_i \rew_{[u.k,
        r_2,\sigma_i]} t_{i+1}$ \\  with $t_{i+1} = t_i[\sigma_i(
      \Solve_{\it JL}\,(\Solve\, (\,\conf{P_1}{\ulm}\,), P_2))\!\!\downarrow_{gql}]_{u.k}
      $. $t_{i+1}$ cannot be reduced at position $u$ because $t_i$ was
      not reducible at position $u$ and the head of the right-hand
      side of rule $r_2$ is the operation $\Solve_{\it JL}$ which
      does not appear in the subterms of the left-hand sides of the
      rules in $\syst$. Therefore,  $t_{i+1}$ can be reducible either at position
      $u.k$ or $u.k.1$ since $t_i$ was reducible at position $u.k$
      only. $t_{i+1}$ is not reducible at position $u.k$ because rule
      $r_3$ cannot be used (no possible pattern-matching). It remains
      position $u.k.1$ at which term $t_{i+1}$ can be reduced if term
      $\sigma_i(P_1)$ is headed by one of the constructors of patterns.
      
    \item[$r = r_3$:] In this case we have
      $t_i \rew_{[u.k, r_3,\sigma_i]} t_{i+1}$ with \\
      $t_{i+1} = t_i[\sigma_i(\Solve_{\it JR}\,(\ulm,
      \Solve\,(\,\conf{P}{\ulm}\,) )\!\!\downarrow_{gql}]_{u.k}$. $t_{i+1}$ cannot be
      reduced at position $u$ because $t_i$ was not reducible at
      position $u$ and the head of the right-hand side of rule $r_3$
      is the operation $\Solve_{\it JR}$ which does not appear in the
      subterms of the left-hand sides of the rules. Therefore,
      $t_{i+1}$ can be reducible either at position $u.k$ or $u.k.2$
      since $t_i$ was reducible at position $u.k$ only. $t_{i+1}$ is
      not reducible at position $u.k$ because rule $r_4$ cannot be
      used (no possible pattern-matching). It remains position $u.k.2$
      at which term $t_{i+1}$ can be reduced if term $\sigma_i(P)$ is
      headed by one of the constructors of patterns.

      \item[$r = r_4$:] In this case we have
      $t_i \rew_{[u.k, r_4,\sigma_i]} t_{i+1}$ with
      $t_{i+1} = t_i[\sigma_i(\conf{\emptypattern}{
        \opn{Join}(\ulm,\ulm')} )\!\!\downarrow_{gql}]_{u.k}$. Notice that the subterm
      $\opn{Join}(\ulm,\ulm')$ is part of a buit-in constraint and is supposed
      to be evaluated (in normal form). Thus the subterm of $t_{i+1}$ at
      position $u.k$ cannot be reducible further. However $t_{i+1}$
      can be reducible at position $u$ by using one of the rules
      $r_0, r_3, r_4, r_6, r_8, r_{10}, r_{12}, r_{13}$.

    \item[$r = r_5$:] In this case we have
      $t_i \rew_{[u.k, r_5,\sigma_i]} t_{i+1}$ with \\
      $t_{i+1} = t_i[\sigma_i(\Solve_{\it
        BI}\,(\Solve\,(\,\conf{P}{\ulm}, \expr, x)) ) \!\!\downarrow_{gql}]_{u.k}
      $. $t_{i+1}$ cannot be reduced at position $u$ because $t_i$ was
      not reducible at position $u$ and the head of the right-hand
      side of rule $r_5$ is the operation $\Solve_{\it BI}$ which does
      not appear in the subterms of the left-hand sides of the
      rules in $\syst$. Therefore, $t_{i+1}$ can be reducible either at position
      $u.k$ or $u.k.1$ since $t_i$ was reducible at position $u.k$
      only. $t_{i+1}$ is not reducible at position $u.k$ because rule
      $r_6$ cannot be used (no possible pattern-matching). It remains
      position $u.k.1$ at which term $t_{i+1}$ can be reduced if term
      $\sigma_i(P)$ is headed by one of the constructors of
      patterns.

      \item[$r=r_6$:] In this case we have
      $t_i \rew_{[u.k, r_4,\sigma_i]} t_{i+1}$ with
      $t_{i+1} = t_i[\sigma_i( \conf{\emptypattern}{\opn{Bind}(\ulm,\expr,x)} )\!\!\downarrow_{gql}]_{u.k}$. Notice that the subterm
      $\opn{Bind}(\ulm,\expr,x)$ is part of a buit-in constraint and is supposed
      to be evaluated (in normal form). Thus, the subterm of $t_{i+1}$ at
      position $u.k$ cannot be reducible further. However $t_{i+1}$
      can be reducible at position $u$ by using one of the rules
      $r_0, r_3, r_4, r_6, r_8, r_{10}, r_{12}, r_{13}$.

    \item[$r = r_7$ :] In this case we have
      $t_i \rew_{[u.k, r_7,\sigma_i]} t_{i+1}$ with \\
      $t_{i+1} = t_i[\sigma_i(\Solve_{\it
        FR}\,(\Solve\,(\,\conf{P}{\ulm}, \expr))\!\!\downarrow_{gql}]_{u.k} $. $t_{i+1}$
      cannot be reduced at position $u$ because $t_i$ was not
      reducible at position $u$ and the head of the right-hand side of
      rule $r_7$ is the operation $\Solve_{\it FI}$ which does not
      appear in the subterms of the left-hand sides of the rules in
      $\syst$. Therefore, $t_{i+1}$ can be reducible either at
      position $u.k$ or $u.k.1$ since $t_i$ was reducible at position
      $u.k$ only. $t_{i+1}$ is not reducible at position $u.k$ because
      rule $r_8$ cannot be used (no possible pattern-matching). It
      remains position $u.k.1$ at which term $t_{i+1}$ can be reduced
      if term $\sigma_i(P)$ is headed by one of the constructors of
      patterns.
      
      \item[$r=r_8$:] In this case we have
      $t_i \rew_{[u.k, r_8,\sigma_i]} t_{i+1}$ with
      $t_{i+1} = t_i[\sigma_i(\conf{\emptypattern}{\opn{Filter}(\ulm,\expr)}  )\!\!\downarrow_{gql}]_{u.k}$. Notice that the subterm
      $\opn{Filter}(\ulm,\expr)$ is part of a buit-in constraint and is supposed
      to be evaluated (in normal form). Thus, the subterm of $t_{i+1}$ at
      position $u.k$ cannot be reducible further. However $t_{i+1}$
      can be reducible at position $u$ by using one of the rules
      $r_0, r_3, r_4, r_6, r_8, r_{10}, r_{12}, r_{13}$.

    \item[$r = r_9$ :] In this case we have
      $t_i \rew_{[u.k, r_9,\sigma_i]} t_{i+1}$ with \\
      $t_{i+1} = t_i[\sigma_i(\Solve_{\it
        BU}(\Solve\,(\,\conf{P}{\ulm}\,), R) )\!\!\downarrow_{gql}]_{u.k} $. $t_{i+1}$
      cannot be reduced at position $u$ because $t_i$ was not
      reducible at position $u$ and the head of the right-hand side of
      rule $r_9$ is the operation $\Solve_{\it BU}$ which does not
      appear in the subterms of the left-hand sides of the rules in
      $\syst$. Therefore, $t_{i+1}$ can be reducible either at
      position $u.k$ or $u.k.1$ since $t_i$ was reducible at position
      $u.k$ only. $t_{i+1}$ is not reducible at position $u.k$ because
      rule $r_{10}$ cannot be used (no possible pattern-matching). It
      remains position $u.k.1$ at which term $t_{i+1}$ can be reduced
      if term $\sigma_i(P)$ is headed by one of the constructors of
      patterns.
      
      \item[$r=r_{10}$:] In this case we have
      $t_i \rew_{[u.k, r_{10},\sigma_i]} t_{i+1}$ with
      $t_{i+1} = t_i[\sigma_i( \conf{\emptypattern}{\opn{Build}(\ulm,R)} )\!\!\downarrow_{gql}]_{u.k}$. Notice that the subterm
      $\opn{Build}(\ulm,R)$ is part of a buit-in constraint and is supposed
      to be evaluated (in normal form). Thus, the subterm of $t_{i+1}$ at
      position $u.k$ cannot be reducible further. However $t_{i+1}$
      can be reducible at position $u$ by using one of the rules
      $r_0, r_3, r_4, r_6, r_8, r_{10}, r_{12}, r_{13}$.

    \item[$r= r_{11}$:] In this case we have
      $t_i \rew_{[u.k, r_{11},\sigma_i]} t_{i+1}$ with \\
      $t_{i+1} = t_i[\sigma_i( \Solve_{\it
        UL}\,(\Solve\,(\,\conf{P_1}{\ulm}\,), P_2) )\!\!\downarrow_{gql}]_{u.k}
      $. $t_{i+1}$ cannot be reduced at position $u$ because $t_i$ was
      not reducible at position $u$ and the head of the right-hand
      side of rule $r_2$ is the operation $\Solve_{\it UL}$ which does
      not appear in the subterms of the left-hand sides of the rules
      in $\syst$. Therefore, $t_{i+1}$ can be reducible either at
      position $u.k$ or $u.k.1$ since $t_i$ was reducible at position
      $u.k$ only. $t_{i+1}$ is not reducible at position $u.k$ because
      rule $r_{12}$ cannot be used (no possible pattern-matching). It
      remains position $u.k.1$ at which term $t_{i+1}$ can be reduced
      if term $\sigma_i(P_1)$ is headed by one of the constructors of
      patterns.
      
    \item[$r = r_{12}$:] In this case we have
      $t_i \rew_{[u.k, r_{12},\sigma_i]} t_{i+1}$ with \\
      $t_{i+1} = t_i[\sigma_i(\Solve_{\it UR}\,(\ulm,
      \Solve\,(\,\conf{P}{\ulm})\,) )\!\!\downarrow_{gql}]_{u.k}$. $t_{i+1}$ cannot be
      reduced at position $u$ because $t_i$ was not reducible at
      position $u$ and the head of the right-hand side of rule $r_{12}$
      is the operation $\Solve_{\it UR}$ which does not appear in the
      subterms of the left-hand sides of the rules. Therefore,
      $t_{i+1}$ can be reducible either at position $u.k$ or $u.k.2$
      since $t_i$ was reducible at position $u.k$ only. $t_{i+1}$ is
      not reducible at position $u.k$ because rule $r_{13}$ cannot be
      used (no possible pattern-matching). It remains position $u.k.2$
      at which term $t_{i+1}$ can be reduced if term $\sigma_i(P)$ is
      headed by one of the constructors of patterns.
      
      \item[$r=r_{13}$:] In this case we have
      $t_i \rew_{[u.k, r_{13},\sigma_i]} t_{i+1}$ with
      $t_{i+1} = t_i[\sigma_i( \conf{\emptypattern}{ \opn{Union}(\ulm,\ulm') } )\!\!\downarrow_{gql}]_{u.k}$. Notice that the subterm
      $ \opn{Union}(\ulm,\ulm') $ is part of a buit-in constraint and is supposed
      to be evaluated (in normal form). Thus, the subterm of $t_{i+1}$ at
      position $u.k$ cannot be reducible further. However $t_{i+1}$
      can be reducible at position $u$ by using one of the rules
      $r_0, r_3, r_4, r_6, r_8, r_{10}, r_{12}, r_{13}$.
      
    \end{description}
\end{proof}

\begin{proposition}[termination]
  \label{prop:terminating}
  The relation $\rew$ is terminating.
\end{proposition}

\begin{proof}
The proof is quite direct. One possible ordering is the
lexicographical ordering $>_{lex}$ defined over terms of the rewriting
system $\syst$ as follows :
$ t_1 > t_2  iff (M(t_1), H(t_1)) >_{lex} (M(t_2), H(t_2))$ where

$H(t)$ is the height of term $t$ and 
$M(t)$ is the multiset including all values $H(P)$ for every patter $P$
occurring in $t$. $>_{lex}$ is we

For rules $r_1, r_2, r_3, r_5, r_7, r_9, r_{11}, r_{12}$, we
have $(M(lhs), H(lhs)) >_{lex} (M(rhs), H(rhs))$ because $M(lhs) >
M(rhs)$. For rules $r_0, r_4, r_6, r_8, r_{10}, r_{13} $  we
have $(M(lhs), H(lhs)) >_{lex} (M(rhs), H(rhs))$ because
$M(lhs) =  M(rhs)$ and $H(lhs) > H(rhs)$.
\end{proof}

In Fig.~\ref{fig:rew-queries}, we enrich the rewriting system $\syst$ by
means of six additional rules which tackle queries as in
Section~\ref{ssec:gral-query}.  Remember that the \emph{result} of a
query $Q$ over a graph $G$ is a graph $\Result_C(Q,G)$ when $Q$ is a
CONSTRUCT query whereas it is a table
$\Result_S(Q,G)$ when
$Q$ is a SELECT query, or both in the case of a $\conselect\ $ query
$\Result_{CS}(Q,G)$. The details of the display functions
$\Print_C$,  $\Print_S$ and $\Print_{CS}$ used in rules $r_{15},
r_{17}$ and $r_{19}$ are omitted as they are out of the scope of the present paper.

For SELECT and \conselect\ queries we use the graph $\Graph(S)$
associated to the set of variables~$S$ (see,
Definition~\ref{def:gral-sem-query-select} or \cite{DEP2021} for
details).

\begin{figure}[!th]
  \caption{$\syst$ (contined): Rewriting rules for queries}
  \label{fig:rew-queries}

\resizebox{\textwidth}{!}{$\displaystyle
   \begin{array}{|llrcl|}
\hline
  &&&& \\
\;r_{14} & : &
\Solve_{Q}(\clause{ CONSTRUCT } R \clause{ WHERE } P, G) & 
  \to & 
  {\it Display}_{C} (R, \Solve(\,\conf{P \clause{ BUILD } R}{\uli_G}\,)\,) \\
\;r_{15} & : &
   {\it Display}_{C} (R, \conf{\emptypattern}{\ulm}) & 
  \to & 
        \Print_C(R, \ulm) \\ 
  &&&& \\
  \;r_{16} & : &
  \Solve_{Q}\,(\clause{ SELECT } S \clause{ WHERE } P, G) &
  \to & 
      {\it Display}_{S} (S, \Solve(\,\conf{P \clause{ BUILD }
        \Graph(S)}{\uli_G}\,)\,) \\
  \;r_{17} & : &
  {\it Display}_{S} (S, \conf{\emptypattern}{\ulm}) &
  \to & 
  \Print_S(S, \ulm) \\
  &&&& \\
  \;r_{18} & : &
  \Solve_{Q}\,(\clause{ \conselect\  } S , R \clause{ WHERE } P, G) &
  \to & 
      {\it Display}_{CS} (S, R, \Solve(\,\conf{P \clause{ BUILD }
        (\Graph(S) \cup R)}{\uli_G}\,)\,) \\
  \;r_{19} & : &
  {\it Display}_{CS} (S, R, \conf{\emptypattern}{\ulm}) &
  \to & 
  \Print_{CS}(S, R, \ulm) \\ 
  &&&& \\
\hline
 \end{array}$
 }
\end{figure}

The Soundness and completeness of the calculus with respect to queries
are direct consequences of Theorems~\ref{th:soundness} and
\ref{th:completeness} and
Definitions~\ref{def:gral-sem-query-construct}, 
\ref{def:gral-sem-query-select} and \ref{def:gral-sem-query-conselect}.

\section{Conclusion and Related Work}
\label{sec:conclusion}

We propose a rule-based calculus for a core graph query language.
The calculus is generic and could easily be adapted to different graph
structures and extended to actual graph query languages. For instance,
graph path variables can be added to the syntax and matches could be
constrained by positive, negative or path constraints just by
performing matches with constraints $Match(L,G,\Phi)$ where $\Phi$
represents constraints in a given logic over items of graphs $L$ and
$G$. To our knowledge, the proposed calculus is the first rule-based
sound and complete calculus dedicated to graph query languages. This
work opens new perspectives regarding the application of verification
techniques to graph-oriented database query languages.  Future work
includes also an implementation of the proposed calculus.

Among related work, we quote first the use of declarative (functional and
logic) languages in the context of relational databases (see,
e.g. \cite{BrasselHanusMueller08PADL,Hanus04JFLP,Almendros-JimenezB03}). 
In these works, the considered databases follow the
relational paradigm which differs from the graph-oriented one that we are
tackling in this paper. Our aim is not to make connections between
graph query languages and functional logic ones. We are rather
interested in investigating formally graph query languages, and
particularly in using dedicated rewriting techniques for such
languages.

The notion of \emph{pattern} present in this paper is close to
the syntactic notions of \emph{clauses} in \cite{cypher} or \emph{graph patterns} in \cite{AnglesABHRV17}. For such syntactic notions, some
authors associate as semantics sets of variables bindings (tables) as
in \cite{cypher,PerezAG09} or simply graphs as in
\cite{AnglesABBFGLPPS18}. In our case, we associate both variable
bindings and graphs since we associate sets of graph homomorphisms to
patterns. This semantics is borrowed from a first work on formal
semantics of graph queries based on category theory  \cite{DEP2021}. Our
semantics allows composition of patterns in a natural way.  Such
composition of patterns is not easy to catch if the semantics is based
only on variable bindings but can be recovered when queries
have graph outcomes as in G-CORE~\cite{AnglesABBFGLPPS18}.



\end{document}